\documentclass[11pt]{article}
\usepackage{amsfonts,amsmath,amsthm}
\usepackage{algorithm2e,framed,algorithmic}
\usepackage{enumerate,bm}

\newcommand{\FNormS}[1]{\mbox{}\left\|#1\right\|_F^2}

\newcommand{\TNorm }[1]{\mbox{}\left\|#1\right\|_2  }
\newcommand{\TNormS}[1]{\mbox{}\left\|#1\right\|_2^2}

\newcommand{\VTTNorm }[1]{\mbox{}\left\|#1\right\|_2  }
\newcommand{\VTTNormS}[1]{\mbox{}\left\|#1\right\|_2^2}

\newcommand{\setlinespacing}[1]%
           {\setlength{\baselineskip}{#1 \defbaselineskip}}

\newcommand{\rank}[1]{{\bf rank}{\left(#1\right)}}
\newcommand{\abs }[1]{\left|#1\right|}

\newtheorem{definition}{Definition}

\newtheorem{lemma}{Lemma}
\newtheorem{theorem}{Theorem}

\newenvironment{Proof}{\noindent {\em Proof:}}{\\\hspace*{\fill}\mbox{$\diamond$}}

\long\def\killtext#1{}

\def\Prob{\hbox{\bf{Pr}}}

\def\Exp{\hbox{\bf{E}}}

\def\math#1{$#1$}

\def\mand#1{$$#1$$}

\def\frac#1#2{{#1\over #2}}

\def\mld#1{\begin{equation}
#1
\end{equation}}
\def\eqar#1{\begin{eqnarray}
#1
\end{eqnarray}}
\def\eqan#1{\begin{eqnarray*}
#1
\end{eqnarray*}}

\def\choose#1#2{\left({{#1}\atop{#2}}\right)}
\def\cl#1{{\cal #1}}

\def\norm#1{{\left\|#1\right\|}}

\def\dotfil{\leaders\hbox to 1.5mm{.}\hfill}
\newcounter{rmnum}
\def\RN#1{\setcounter{rmnum}{#1}\uppercase\expandafter{\romannumeral\value{rmnum}}}
\def\rn#1{\setcounter{rmnum}{#1}\expandafter{\romannumeral\value{rmnum}}}

\newcommand\remove[1]{}
\DeclareMathSymbol{\R}{\mathbin}{AMSb}{"52}

\newlength{\defbaselineskip}
\begin{document}

\title{
Fast approximation of matrix coherence and statistical leverage
}

\author{
Petros Drineas
\thanks{
Dept. of Computer Science,
Rensselaer Polytechnic Institute,
Troy, NY 12180.
Email: drinep@cs.rpi.edu
}
\and
Malik Magdon-Ismail
\thanks{
Dept. of Computer Science,
Rensselaer Polytechnic Institute,
Troy, NY 12180.
Email: magdon@cs.rpi.edu
}
\and
Michael W. Mahoney
\thanks{
Dept. of Mathematics,
Stanford University,
Stanford, CA 94305.
Email: mmahoney@cs.stanford.edu
}
\and
David P. Woodruff
\thanks{
IBM Almaden Research Center,
650 Harry Road,
San Jose, CA 95120 USA.
Email: dpwoodru@us.ibm.com
}
}

\date{}
\maketitle


\begin{abstract}
The \emph{statistical leverage scores} of a matrix $A$ are the squared row-norms of the matrix containing its (top) left singular vectors and the \emph{coherence} is the largest leverage score.
These quantities are of interest in recently-popular problems such as 
matrix completion and Nystr\"{o}m-based low-rank matrix approximation as 
well as in large-scale statistical data analysis applications more 
generally; moreover, they are of interest since they define the key 
structural nonuniformity that must be dealt with in developing fast 
randomized matrix algorithms. 
Our main result is a randomized algorithm that takes as input an arbitrary $n \times d$ matrix $A$, with $n \gg d$, and that returns as output relative-error approximations to \emph{all} $n$ of the statistical leverage scores. The proposed algorithm runs (under assumptions on the precise values of $n$ and $d$) in $O(n d \log n)$  time, as opposed to the $O(nd^2)$ time required by the na\"{i}ve algorithm that involves computing an orthogonal basis for the range of~\math{A}. Our analysis may be viewed in terms of computing a relative-error approximation to an \emph{under}constrained least-squares approximation problem, or, relatedly, it may be viewed as an application of Johnson-Lindenstrauss type ideas. Several practically-important extensions of our basic result are also described, including the approximation of so-called cross-leverage scores, the extension of these ideas to matrices with $n \approx d$, and the extension to streaming environments.
\end{abstract} 

\section{Introduction}
\label{sxn:intro}

The concept of \emph{statistical leverage} measures the extent to which the
singular vectors of a matrix are correlated with the standard basis and
as such it has found usefulness recently in large-scale data analysis and
in the analysis of randomized matrix
algorithms~\cite{VW81,CUR_PNAS,DMM08_CURtheory_JRNL}.
A related notion is that of \emph{matrix coherence}, which has been of
interest in recently popular problems such as matrix completion and
Nystr\"{o}m-based low-rank matrix approximation~\cite{CR08_TR,TalRos10}.
Defined more precisely below, the statistical leverage scores may be
computed as the squared Euclidean norms of the rows of the matrix
containing the top left singular vectors and the coherence of the matrix is
the largest statistical leverage score.
Statistical leverage scores have a long history in statistical data
analysis, where they have been used for outlier detection in regression
diagnostics~\cite{HW78,CH86}.
Statistical leverage scores have also proved crucial recently in the
development of improved worst-case randomized matrix algorithms that are
also amenable to high-quality numerical implementation and that are useful
to domain scientists~\cite{DMM08_CURtheory_JRNL,CUR_PNAS,BMD09_CSSP_SODA,DMM06,Sarlos06,DMMS07_FastL2_NM10};
see~\cite{Mah-mat-rev_BOOK} for a detailed discussion.
The na\"{i}ve and best previously existing algorithm to compute these scores
would compute an orthogonal basis for the dominant part of the spectrum of
$A$, \emph{e.g.}, the basis provided by the Singular Value Decomposition (SVD)
or a basis provided by a QR decomposition~\cite{GVL96},
and then use that basis to compute diagonal elements of the projection
matrix onto the span of that basis.

We present a randomized algorithm to compute relative-error approximations
to every statistical leverage score in time qualitatively faster than the
time required to compute an orthogonal basis.
For the case of an arbitrary $n \times d$ matrix $A$, with $n \gg d$, our
main algorithm runs (under assumptions on the precise values of $n$ and $d$, see Theorem~\ref{thm:main_result} for an exact statement) in $O( n d \log n / \epsilon^2)$ time, as
opposed to the $\Theta(nd^2)$ time required by the na\"{i}ve algorithm.
As a corollary, our algorithm provides a relative-error approximation to the
coherence of an arbitrary matrix in the same time.
In addition, several practically-important extensions of the basic idea
underlying our main algorithm are also described in this paper.

\subsection{Overview and definitions}

We start with the following definition of the \emph{statistical leverage
scores} of a matrix.
\begin{definition}
\label{def:lev-scores}
Given an arbitrary $n \times d$ matrix $A$, with $n > d$, let $U$ denote
the $n \times d$ matrix consisting of the $d$ left singular vectors of $A$,
and let $U_{(i)}$ denote the $i$-th row of the matrix $U$ as a row vector.
Then, the \emph{statistical leverage scores} of the rows of $A$ are given by
\begin{equation}
\ell_i = \VTTNormS{U_{(i)}} ,
\end{equation}
for $i\in\{1,\ldots,n\}$;
the \emph{coherence} $\gamma$ of the rows of $A$ is
\begin{equation}
\gamma = \max_{i\in\{1,\ldots,n\}} \ell_i  ,
\end{equation}
i.e., 
it is the largest statistical leverage score of $A$;
and the \emph{$(i,j)$-cross-leverage scores} $c_{ij}$ are
\begin{equation}
c_{ij} = \left\langle U_{(i)}, U_{(j)} \right\rangle ,
\end{equation}
i.e., 
they are the dot products between the $i^{th}$ row and the $j^{th}$ row of $U$.
\end{definition}

\noindent Although we have defined these quantities in terms of a particular basis, they clearly do not depend on that particular basis, but
only on the space spanned by that basis. To see this, let $P_A$ denote the projection matrix onto the span of the columns of $A$.
Then,
\begin{equation}
\ell_i = \VTTNormS{U_{(i)}}
       = \left(UU^T\right)_{ii}
       = \left(P_A\right)_{ii} .
\label{eqn:lev}
\end{equation}
That is, the statistical leverage scores of a matrix $A$ are equal to the
diagonal elements of the projection matrix onto the span of its columns.%
\footnote{In this paper, for simplicity of exposition, we consider the case 
that the matrix $A$ has rank equal to $d$, \emph{i.e.}, has full column 
rank.  Theoretically, the extension to rank-deficient matrices $A$ is 
straightforward---simply modify Definition~\ref{def:lev-scores} and thus 
Eqns.~(\ref{eqn:lev}) and~(\ref{eqn:crosslev}) to let $U$ be any orthogonal 
matrix (clearly, with fewer than $d$ columns) spanning the column space of 
$A$.  From a numerical perspective, things are substantially more subtle, 
and we leave this for future work that considers numerical implementations 
of our algorithms.}
Similarly, the $(i,j)$-cross-leverage scores are equal to the off-diagonal
elements of this projection matrix, \emph{i.e.},
\begin{equation}
c_{ij} = (P_A)_{ij} = \left\langle U_{(i)}, U_{(j)} \right\rangle  .
\label{eqn:crosslev}
\end{equation}
Clearly, $O(nd^2)$ time suffices to compute all the statistical leverage
scores exactly: simply perform the SVD or compute a QR decomposition of $A$
in order to obtain \emph{any} orthogonal basis for the range of $A$ and
then compute the Euclidean norm of the rows of the resulting matrix.
Thus, in this paper, we are interested in algorithms that run in $o(nd^2)$
time.

Several additional comments are worth making regarding this definition.
First, since $\sum_{i=1}^{n}\ell_i=\FNormS{U}=d$, we can define
a probability distribution over the rows of $A$ as \math{p_i=\ell_i/d}.
As discussed below, these probabilities have played an important role in
recent work on randomized matrix algorithms and an important algorithmic
question is the degree to which they are uniform or nonuniform.%
\footnote{Observe that if $U$ consists of $d$ columns from the identity,
then the leverage scores are extremely nonuniform: $d$ of them are equal to one
and the remainder are equal to zero. On the other hand, if $U$ consists of $d$ columns from a
normalized Hadamard transform (see Section~\ref{sxn:RHT} for a definition),
then the leverage scores are very uniform: all $n$ of them are equal to $d/n$.}
Second, one could also define leverage scores for the columns of a ``tall'' matrix
$A$, but clearly those are all equal to one unless $n<d$ or $A$ is
rank-deficient. Third, and more generally, given a rank parameter $k$, one can define the
\emph{statistical leverage scores relative to the best rank-$k$
approximation to $A$} to be the $n$ diagonal
elements of the projection matrix onto the span of $A_k$, the best rank-$k$
approximation to~$A$.

\subsection{Our main result}

Our main result is a randomized algorithm for computing relative-error approximations to every statistical leverage score, as well as an additive-error approximation to all of the large cross-leverage scores, of an arbitrary $n \times d$ matrix, with $n \gg d$, in time qualitatively faster
than the time required to compute an orthogonal basis for the range of that
matrix. Our main algorithm for computing approximations to the statistical leverage scores (see Algorithm~\ref{alg:Leverage_Scores} in
Section~\ref{sxn:mainalg}) will amount to constructing a ``randomized sketch'' of the input matrix and then computing the Euclidean norms of the
rows of that sketch. This sketch can also be used to compute approximations to the large cross-leverage scores (see Algorithm~\ref{alg:Off_Diagonal} of Section~\ref{sxn:mainalg}).

The following theorem provides our main quality-of-approximation and running time result for Algorithm~\ref{alg:Leverage_Scores}.
\begin{theorem}
\label{thm:main_result}
Let $A$ be a full-rank $n \times d$ matrix, with $n \gg d$; let $\epsilon \in (0,1/2]$ be an error parameter; and recall the definition of the statistical leverage scores $\ell_i$ from
Definition~\ref{def:lev-scores}. Then, there exists a randomized algorithm (Algorithm~\ref{alg:Leverage_Scores} of Section~\ref{sxn:mainalg} below) that returns values $\tilde{\ell}_i$, for all $i \in \{1,\ldots,n\}$, such that with probability at least~$0.8$,
\begin{equation}
\abs{\ell_i - \tilde{\ell}_i} \leq \epsilon \ell_i
\end{equation}
holds for all $i \in \{1,\ldots,n\}$. Assuming $d \leq n \leq e^d$, the running time of the algorithm is
\mand{O\left(nd\ln\left(d\epsilon^{-1}\right) + nd\epsilon^{-2}\ln n + d^3\epsilon^{-2}\left(\ln n\right)\left(\ln \left(d\epsilon^{-1}\right)\right)\right).}
\end{theorem}
\noindent Algorithm~\ref{alg:Leverage_Scores} provides a relative-error approximation to all of the statistical leverage scores $\ell_i$ of~$A$ and, assuming \math{d \ln d=o\left(\frac{n}{\ln n}\right)}, \math{\ln n=o\left(d\right)}, and treating $\epsilon$ as a constant, its running time is \math{o(nd^2)}, as desired. 
As a corollary, the largest leverage score (and thus the coherence) is approximated to relative-error in $o(nd^2)$ time.

The following theorem provides our main quality-of-approximation and running time result for Algorithm~\ref{alg:Off_Diagonal}.
\begin{theorem} \label{thm:main_result-b}
Let $A$ be a full-rank $n \times d$ matrix, with $n \gg d$; let $\epsilon \in (0,1/2]$ be an error parameter; let $\kappa$ be a parameter; and recall the definition of the cross-leverage scores $c_{ij}$ from Definition~\ref{def:lev-scores}. Then, there exists a randomized algorithm (Algorithm~\ref{alg:Off_Diagonal} of Section~\ref{sxn:mainalg} below) that returns the pairs \math{\{(i,j)\}} together with estimates \math{\{\tilde c_{ij}\}} such that, with probability at least~$0.8$,
\begin{enumerate}[i.]
\item If \math{\displaystyle c_{ij}^2\ge \frac{d}{\kappa}+12\epsilon\ell_i\ell_j}, then \math{(i,j)} is returned; if \math{(i,j)} is  returned, then \math{\displaystyle
c_{ij}^2\ge \frac{d}{\kappa}-30\epsilon\ell_i\ell_j}.
\item For all pairs \math{(i,j)} that are returned, \math{\tilde c_{ij}^2-30\epsilon\ell_i\ell_j\le c_{ij}^2\le \tilde c_{ij}^2+12\epsilon\ell_i\ell_j.}
\end{enumerate}
This algorithm runs in \math{O(\epsilon^{-2}n\ln n+\epsilon^{-3}\kappa d\ln^2 n)} time.
\end{theorem}
\noindent Note that by setting \math{\kappa=n\ln n}, we can compute all the ``large'' cross-leverage scores, \emph{i.e.}, those satisfying \math{c_{ij}^2 \ge \frac{d}{n\ln n}}, to within additive-error in \math{O\left(nd\ln^3 n\right)} time (treating $\epsilon$ as a constant). If \math{\ln^3 n=o\left(d\right)} the overall running time is $o(nd^2)$, as desired.

\subsection{Significance and related work}

Our results are important for their applications to fast randomized
matrix algorithms, as well as their applications in numerical linear algebra
and large-scale data analysis more generally.

\textbf{Significance in theoretical computer science.}
The statistical leverage scores define the key structural nonuniformity
that must be dealt with (\emph{i.e.}, either rapidly approximated or rapidly
uniformized at the preprocessing step) in developing fast randomized
algorithms for matrix problems such as least-squares
regression~\cite{Sarlos06,DMMS07_FastL2_NM10} and low-rank matrix
approximation~\cite{PRTV00,Sarlos06,DMM08_CURtheory_JRNL,CUR_PNAS,BMD09_CSSP_SODA}.
Roughly, the best random sampling algorithms use these scores (or the
generalized leverage scores relative to the best rank-$k$ approximation to
$A$) as an importance sampling distribution to sample with respect to. On the other hand, the best random projection algorithms rotate to a basis where these
scores are approximately uniform and thus in which uniform sampling is
appropriate.
See~\cite{Mah-mat-rev_BOOK} for a detailed discussion.

As an example, the CUR decomposition of~\cite{DMM08_CURtheory_JRNL,CUR_PNAS}
essentially computes $p_i = \ell_i/k$, for all $i\in \left\{1,\ldots,n\right\}$ and for a rank
parameter $k$, and it uses these as an importance sampling distribution.
The computational bottleneck for these and related random sampling algorithms is the
computation of the importance sampling probabilities.
On the other hand, the
computational bottleneck for random projection algorithms is the
application of the random projection, which is sped up by using variants of
the Fast Johnson-Lindenstrauss Transform~\cite{AC06,AC06-JRNL09}. 
By our main result, the leverage scores (and thus these probabilities) can be
approximated in time that depends on an application of a Fast
Johnson-Lindenstrauss Transform.
In particular, the random sampling algorithms
of~\cite{DMM06,DMM08_CURtheory_JRNL,CUR_PNAS} for least-squares
approximation and low-rank matrix approximation now run in time that is
essentially the same as the best corresponding random projection algorithm
for those problems~\cite{Sarlos06}.

\textbf{Applications to numerical linear algebra.}
Recently, high-quality numerical implementations of variants of the basic
randomized matrix algorithms have proven superior to traditional
deterministic algorithms~\cite{RT08,RST09,AMT10}.
An important question raised by our main results is how these will compare
with an implementation of our main algorithm.
More generally, density functional theory~\cite{BKS07} and uncertainty
quantification~\cite{BCF10} are two scientific computing areas where
computing the diagonal elements of functions (such as a projection or
inverse) of very large input matrices is common.
For example, in the former case, ``heuristic'' methods based on using
Chebychev polynomials have been used in numerical linear algebra to compute
the diagonal elements of the projector~\cite{BKS07}.
Our main algorithm should have implications in both of these areas.

\textbf{Applications in large-scale data analysis.}
The statistical leverage scores and the scores relative to the best rank-$k$
approximation to $A$ are equal to the diagonal elements of the so-called
``hat matrix''~\cite{HW78,ChatterjeeHadi88}.
As such, they have a natural statistical interpretation in terms of the
``leverage'' or ``influence'' associated with each of the data
points~\cite{HW78,CH86,ChatterjeeHadi88}.
In the context of regression problems, the $i^{th}$ leverage score
quantifies the leverage or influence of the $i^{th}$ constraint/row of~$A$
on the solution of the overconstrained least squares optimization problem
$ \min_x \TNorm{Ax-b}$
and the $(i,j)$-th cross leverage score quantifies how much influence or
leverage the $i^{th}$ data point has on the $j^{th}$ least-squares
fit (see~\cite{HW78,CH86,ChatterjeeHadi88} for details).
When applied to low-rank matrix approximation problems, the leverage score
$\ell_{j}$ quantifies the amount of leverage or influence exerted by the
$j^{th}$ column of $A$ on its optimal low-rank approximation.
Historically, these quantities have been widely-used for outlier
identification in diagnostic regression
analysis~\cite{VW81,ChatterjeeHadiPrice00}.

More recently, these scores (usually the largest scores) often have an
interpretation in terms of the data and processes generating the data that
can be exploited.
For example, depending on the setting, they can have an interpretation in
terms of high-degree nodes in data graphs, very small clusters in noisy
data, coherence of information, articulation points between clusters, the
value of a customer in a network, space localization in sensor networks,
etc.~\cite{Bon87,RD02,newman05_betweenness,JLHB07,Mah-mat-rev_BOOK}.
In genetics, dense matrices of size thousands by hundreds of thousands
(a size scale at which even traditional deterministic QR algorithms fail
to run) constructed from DNA Single Nucleotide Polymorphisms (SNP) data are
increasingly common, and the statistical leverage scores can correlate
strongly with other metrics of genetic interest~\cite{Paschou07b,CUR_PNAS,Paschou10a,Paschou10b}.
Our main result will permit the computation of these scores and related
quantities for significantly larger SNP data sets than has been possible
previously~\cite{Paschou07b,Paschou10a,Paschou10b,Sayan11-unpub}.

\textbf{Remark.}
Lest there be any confusion, we should emphasize our main contributions.
First, note that statistical leverage and matrix coherence are important 
concepts in statistics and machine learning.
Second, recall that several random sampling algorithms for ubiquitous matrix 
problems such as least-squares approximation and low-rank matrix 
approximation use leverage scores in a crucial manner; but until now these 
algorithms were $\Omega(T_{SVD})$, where $T_{SVD}$ is the time required to 
compute a QR decomposition or a partial SVD of the input matrix.
Third, note that, in some cases, $o(T_{SVD})$ algorithms exist for these 
problems based on fast random projections.
But recall that the existence of 
those projection algorithms \emph{in no way implies} that it is easy or 
obvious how to compute the statistical leverage scores efficiently.
Fourth, one implication of our main result is that those random sampling 
algorithms can now be performed \emph{just as efficiently} as those 
random projection algorithms; thus, the solution for those matrix problems 
can now be obtained while preserving the identity of the rows.
That is, 
these problems can now be solved just as efficiently by using actual rows, 
rather than the arbitrary linear combinations of rows that are returned by 
random projections.
Fifth, we provide a generalization to ``fat'' matrices and to obtaining the 
cross-leverage scores.
Sixth, we develop algorithms that can compute leverage scores and related 
statistics even in streaming environments.

\subsection{Empirical discussion of our algorithms}
\label{sxn:intro-empirical}

Although the main contribution of our paper is to provide a rigorous 
theoretical understanding of fast leverage score approximation, our paper 
does analyze the theoretical performance of what is meant to be a 
practical algorithm.
Thus, one might wonder about the empirical performance of our 
algorithms---for example, whether hidden constants render the algorithms
useless for data of realistic size.
Not surprisingly, this depends heavily on the quality of the numerical 
implementation, whether one is interested in ``tall'' or more general 
matrices, etc.
We will consider empirical and numerical aspects of these algorithms in 
forthcoming papers, \emph{e.g.},~\cite{GM12_InPrep}.
We will, however, provide here a brief summary of several numerical issues
for the reader interested in these issues.

Empirically, the running time bottleneck for our main algorithm 
(Algorithm~\ref{alg:Leverage_Scores} of Section~\ref{sxn:mainalg}) applied 
to ``tall'' matrices is the application of the random projection $\Pi_1$.
Thus, empirically the running time is similar to the running time of random 
projection based methods for computing approximations to the least-squares
problem, which is also dominated by the application of the random 
projection.
The state of the art here is the Blendenpik algorithm of~\cite{AMT10} and
the LSRN algorithm of~\cite{MSM11_TR}.
In their Blendenpik paper, Avron, Maymounkov, and Toledo showed that their 
high-quality numerical implementation of a Hadamard-based random projection 
(and associated least-squares computation) ``beats \textsc{Lapack}'s%
\footnote{\textsc{Lapack} (short for Linear Algebra PACKage) is a
high-quality and widely-used software library of numerical routines for
solving a wide range of numerical linear algebra problems.}
direct dense least-squares solver by a large margin on essentially any 
dense tall matrix,'' and they concluded that their empirical results 
``suggest that random projection algorithms should be incorporated into 
future versions of \textsc{Lapack}''~\cite{AMT10}. 
The LSRN algorithm of Meng, Saunders, and Mahoney improves Blendenpik in 
several respects, \emph{e.g.}, providing better handling of sparsity and 
rank deficiency, but most notably the random projection underlying LSRN is 
particularly appropriate for solving large problems on clusters with high 
communication cost, \emph{e.g.}, it has been shown to scale well on Amazon 
Elastic Cloud Compute clusters.  
Thus, our main algorithm should extend easily to these environments with 
the use of the random projection underlying LSRN.
Moreover, for both Blendenpik and LSRN (when implemented with a 
Hadamard-based random projection), the hidden constants in the 
Hadamard-based random projection are so \emph{small} that the random 
projection algorithm (and thus the empirical running time of our main 
algorithm for approximating leverage scores) beats the traditional 
$O(nd^2)$ time algorithm for dense matrices as small as thousands of rows 
by hundreds of columns.

\subsection{Outline}

In Section~\ref{sxn:review}, we will provide a brief review of relevant notation and concepts from linear algebra. Then, in Sections~\ref{sxn:mainalg} and~\ref{sxn:mainproof}, we will present our main results: Section~\ref{sxn:mainalg} will contain our main algorithm and Section~\ref{sxn:mainproof} will contain the proof of our main theorem. Section~\ref{sxn:extensions} will then describe extensions of our main result to general ``fat'' matrices, \emph{i.e.}, those with $n \approx d$. Section~\ref{sxn:discussion} will conclude by describing the relationship of our main result with another related estimator for the statistical leverage scores, an application of our main algorithm to the under-constrained least-squares approximation problem, and extensions of our main algorithm to streaming environments. 

\section{Preliminaries on linear algebra and fast random projections}
\label{sxn:review}

\subsection{Basic linear algebra and notation}

Let $[n]$ denote the set of integers $\{1,2,\ldots,n\}$.
For any matrix $A \in \mathbb{R}^{n \times d}$, let $A_{(i)}$, $i \in [n]$,
denote the $i$-th row of $A$ as a row vector, and let $A^{(j)}$,
$j \in [d]$ denote the $j$-th column of $A$ as a column vector.
Let $\FNormS{A} = \sum_{i=1}^n \sum_{j=1}^d A_{ij}^2$ denote the square
of the Frobenius norm of $A$, and let
$\TNorm{A} = \sup_{\ \VTTNorm{x}=1} \VTTNorm{Ax}$ denote the spectral norm
of $A$.
Relatedly, for any vector $x \in \mathbb{R}^n$, its Euclidean norm (or
$\ell_2$-norm) is the square root of the sum of the squares of its elements.
The dot product between two vectors $x,y\in\mathbb{R}^{n}$ will be denoted
$\langle x,y \rangle$, or alternatively as $x^Ty$.
Finally, let $e_i \in \mathbb{R}^{n}$, for all $i \in [n]$, denote the
standard basis vectors for $\mathbb{R}^{n}$ and let $I_n$ denote the $n \times n$ identity matrix.

Let the rank of $A$ be $\rho \leq \min\{n,d\}$, in which case the
``compact'' or ``thin''
SVD of $A$ is denoted by $ A = U \Sigma V^T $,
where $U \in \mathbb{R}^{n \times \rho}$,
$\Sigma \in \mathbb{R}^{\rho \times \rho}$, and
$V \in \mathbb{R}^{d \times \rho}$.
(For a general matrix $X$, we will write $X=U_X\Sigma_XV_X^T$.)
Let $\sigma_i(A), i\in [\rho]$ denote the $i$-th singular value of $A$, and
let $\sigma_{\max}(A)$ and $\sigma_{\min}(A)$ denote the maximum and
minimum singular values of $A$, respectively. The Moore-Penrose pseudoinverse of $A$ is the $d \times n$ matrix defined
by $ A^{\dagger} = V \Sigma^{-1} U^T $~\cite{Nashed76}.
Finally, for any orthogonal matrix $U \in \mathbb{R}^{n \times \ell}$, let
$U^{\perp} \in \mathbb{R}^{n \times (n-\ell)}$ denote an orthogonal matrix
whose columns are an orthonormal basis spanning the subspace of
$\mathbb{R}^{n}$ that is orthogonal to the subspace spanned by the columns
of~$U$ (\emph{i.e.}, the range of $U$).
It is always possible to extend an orthogonal matrix $U$ to a full
orthonormal basis of $\mathbb{R}^{n}$ as $[U \quad U^{\perp} ]$.

The SVD is important for a number of reasons~\cite{GVL96}.
For example, the projection of the columns of $A$ onto the $k$ left
singular vectors associated with the top $k$ singular values gives the best
rank-$k$ approximation to $A$ in the spectral and Frobenius norms.
Relatedly, the solution to the least-squares (LS) approximation problem is
provided by the SVD:
given an $n \times d$ matrix $A$ and an $n$-vector $b$, the LS problem is to
compute the minimum $\ell_2$-norm vector $x$ such that
$
\VTTNorm{Ax-b}
$
is minimized over all vectors $x\in\mathbb{R}^{d}$. This optimal vector is given by $x_{opt}=A^{\dagger}b$.
We call a LS problem \emph{overconstrained} (or \emph{overdetermined}) if
$n>d$ and \emph{underconstrained} (or \emph{underdetermined}) if $n<d$.

\subsection{The Fast Johnson-Lindenstrauss Transform (FJLT)}
\label{sxn:FJLT}

Given \math{\epsilon>0} and a set of points
\math{x_1,\ldots,x_n} with \math{x_i\in\R^d},
a \math{\epsilon}-Johnson-Lindenstrauss Transform (\math{\epsilon}-JLT),
denoted
\math{\Pi\in\R^{r\times d}},  is a projection
of the points into \math{\R^r} such that
\begin{equation}
(1-\epsilon)\norm{x_i}_2^2\le \norm{\Pi x_i}_2^2\le
(1+\epsilon)\norm{x_i}_2^2.
\end{equation}
To construct an \math{\epsilon}-JLT with high probability, simply choose
every entry of \math{\Pi} independently, equal to \math{\pm\sqrt{3/r}} with
probability \math{1/6} each and zero otherwise (with probability
\math{2/3})~\cite{Ach03_JRNL}.
Let~\math{\Pi_{JLT}} be a matrix drawn from such a distribution over
\math{r\times d} matrices.%
\footnote{When no confusion can arise, we will use \math{\Pi_{JLT}}
to refer to this distribution over matrices as well as to
a specific matrix drawn from this distribution.} 
Then, the following lemma holds.
\begin{lemma}[Theorem 1.1 of~\cite{Ach03_JRNL}]
\label{lemma:eJLT}
Let \math{x_1,\ldots,x_n} be an arbitrary
(but fixed) set of points, where \math{x_i\in\R^d} and let $0<\epsilon\leq 1/2$ be an accuracy parameter. If
\mand{r\ge \frac{1}{\epsilon^{2}}\left(12\ln n +6\ln\frac{1}{\delta}\right)}
then, with probability at least \math{1-\delta},
\math{\Pi_{JLT}\in\R^{r\times d}} is an \math{\epsilon}-JLT .
\end{lemma}
\noindent For our main results, we will also need a stronger requirement than the simple \math{\epsilon}-JLT and so we will use a version of the Fast Johnson-Lindenstrauss Transform (FJLT), which was originally introduced in~\cite{AC06,AC06-JRNL09}. Consider an orthogonal matrix \math{U\in\R^{n\times d}}, viewed as \math{d} vectors in \math{\R^n}. A FJLT projects the vectors from \math{\R^n} to \math{\R^r}, while preserving the orthogonality of \math{U}; moreover, it does so very quickly. Specifically, given \math{\epsilon>0}, \math{\Pi\in\R^{r\times n}} is an \math{\epsilon}-FJLT for \math{U} if
\begin{itemize}
\item \math{\norm{I_d-U^T\Pi^T\Pi U}_2\le\epsilon}, and
\item given any \math{X\in\R^{n\times d}}, the matrix product \math{\Pi X} can be computed in \math{O(nd\ln r)} time.
\end{itemize}
The next lemma follows from the definition of an \math{\epsilon}-FJLT, and its proof can be found in~\cite{DMM06,DMMS07_FastL2_NM10}.
\begin{lemma}\label{lemma:FJLT}
Let \math{A} be any matrix in $\mathbb{R}^{n \times d}$ with $n \gg d$ and $\rank{A} = d$. Let the SVD of $A$ be $A=U\Sigma V^T$, let \math{\Pi} be an \math{\epsilon}-FJLT for \math{U} (with
\math{0<\epsilon\leq 1/2}) and let \math{\Psi=\Pi U=U_{\Psi} \Sigma_{\Psi}V^T_{\Psi}}. Then, all the following hold:
\eqar{\rank{\Pi A}&=&\rank{\Pi U}=\rank{U}=\rank{A}=d,\label{lemma:FJLTa}\\
\norm{I-\Sigma_{\Psi}^{-2}}_2&\le&
\epsilon/(1-\epsilon),\text{\ and}\label{lemma:FJLTb}\\
(\Pi A)^\dagger&=&V\Sigma^{-1}(\Pi U)^\dagger .
\label{lemma:FJLTc}
}
\end{lemma}

\subsection{The Subsampled Randomized Hadamard Transform (SRHT)}
\label{sxn:RHT}

One can use a Randomized Hadamard Transform (RHT) to construct, with high
probability, an \math{\epsilon}-FJLT. Our main algorithm will use this efficient construction in a crucial way.%
\footnote{Note that the RHT has also been crucial in the development of
$o(nd^2)$ randomized algorithms for the general overconstrained LS
problem~\cite{DMMS07_FastL2_NM10} and its variants have been used to
provide high-quality numerical implementations of such
randomized algorithms~\cite{RT08,AMT10}.}
Recall that the (unnormalized) $n \times n$ matrix of the Hadamard
transform $\hat H_n$ is defined recursively by
$$ \hat H_{2n} = \left[
\begin{array}{cc}
  \hat H_{n} & \hat H_{n} \\
 \hat  H_{n} & -\hat H_{n}
\end{array}\right] ,
$$
with \math{\hat H_1=1}. The $n \times n$ normalized matrix of the Hadamard transform is equal to $$H_n=\hat H_n/\sqrt{n}.$$ From now on, for simplicity and without loss of generality, we assume that \math{n} is a power of 2 and we will suppress \math{n} and just write \math{H}. 
(Variants of this basic construction that relax this assumption and that are
more appropriate for numerical implementation have been described and
evaluated in~\cite{AMT10}.)
Let $D \in \mathbb{R}^{n \times n}$ be a random diagonal matrix with independent diagonal entries $D_{ii}=+1$ with probability $1/2$ and $D_{ii}=-1$ with probability $1/2$. The product $HD$ is a RHT and it has three useful properties. First, when applied to a vector, it ``spreads out'' its energy. Second, computing the product $HDx$ for any vector $x \in \mathbb{R}^n$ takes $O(n\log_2 n)$ time. Third, if we only need to access $r$ elements in the transformed vector, then those $r$ elements can be computed in $O(n \log_2 r )$ time~\cite{AL08}. The Subsampled Randomized Hadamard Transform (SRHT) randomly samples (according to the uniform distribution) a set of \math{r} rows of a RHT.

Using the sampling matrix formalism described previously~\cite{dkm_matrix1,DMM06,DMM08_CURtheory_JRNL,DMMS07_FastL2_NM10}, we will represent the operation of randomly sampling $r$ rows of an
$n \times d$ matrix $A$ using an $r \times n$ linear sampling operator $S^T$. Let the  matrix $\Pi_{FJLT}=S^THD$ be generated using the SRHT.%
\footnote{Again, when no confusion can arise, we will use $\Pi_{FJLT}$ to denote a specific SRHT or the distribution on matrices implied by the randomized process for constructing an SRHT.} 
The most important property about the distribution $\Pi_{FJLT}$ is that if \math{r} is large enough, then, with high probability, $\Pi_{FJLT}$ generates an~\math{\epsilon}-FJLT. We summarize this discussion in the following lemma (which is essentially a combination of Lemmas~3 and~4 from~\cite{DMMS07_FastL2_NM10}, restated to fit our notation).
\begin{lemma}\label{thm:theorem7correct}
Let \math{\Pi_{FJLT} \in \mathbb{R}^{r \times n}} be generated using the SRHT as described above and let \math{U\in\R^{n\times d}} ($n\gg d$) be an (arbitrary but fixed) orthogonal matrix. If
\mand{r \geq \frac{14^2d\ln(40nd)}{\epsilon^2}\ln \left(\frac{30^2 d\ln(40nd)}{\epsilon^2}\right),}
then, with probability at least $0.9$, \math{\Pi_{FJLT}} is an \math{\epsilon}-FJLT for \math{U}.
\end{lemma}

\section{Our main algorithmic results}
\label{sxn:mainalg}

In this section, we will describe our main results for computing relative-error approximations to every statistical leverage score (see Algorithm~\ref{alg:Leverage_Scores}) as well as additive-error approximations to all of the large cross-leverage scores (see Algorithm~\ref{alg:Off_Diagonal}) of an arbitrary matrix $A \in \mathbb{R}^{n \times d}$, with $n \gg d$. Both algorithms make use of a ``randomized sketch'' of $A$ of the form $A(\Pi_1 A)^{\dagger}\Pi_2  $, where $\Pi_1$ is an \math{\epsilon}-FJLT and $\Pi_2$ is an \math{\epsilon}-JLT. We start with a high-level description of the basic ideas.

\subsection{Outline of our basic approach}
\label{sxn:mainalg-basic}

Recall that our first goal is to
approximate, for all $i \in [n]$, the quantities
\begin{equation}
\label{eqn:pidef-itoU}
\ell_i = \TNormS{U_{(i)}}
    = \TNormS{e_i^T U}   ,
\end{equation}
where $e_i$ is a standard basis vector.
The hard part of computing the scores $\ell_i$ according to
Eqn.~(\ref{eqn:pidef-itoU}) is computing an orthogonal matrix $U$ spanning
the range of $A$, which takes $O(nd^2)$ time.
Since $U U^T = A A^{\dagger}$,
it follows that
\begin{equation}
\label{eqn:pidef-itoA}
\ell_i = \TNormS{e_i^T U U^T}
    = \TNormS{e_i^T A A^{\dagger}}
    = \TNormS{ ( A A^{\dagger} )_{(i)} }   ,
\end{equation}
where the first equality follows from the orthogonality of (the columns of)
$U$. The hard part of computing the scores $\ell_i$ according to
Eqn.~(\ref{eqn:pidef-itoA}) is two-fold: first, computing the pseudoinverse;
and second, performing the matrix-matrix multiplication of $A$ and
$A^{\dagger}$.
Both of these procedures take $O(nd^2)$~time. As we will see, we can get around both of these bottlenecks by the judicious application of random projections to Eqn.~(\ref{eqn:pidef-itoA}).

To get around the bottleneck of $O(nd^2)$ time due to computing \math{A^\dagger} in Eqn.~(\ref{eqn:pidef-itoA}), we will compute the pseudoinverse of a ``smaller'' matrix that approximates $A$.
A necessary condition for such a smaller matrix is that it preserves rank. So, na\"{i}ve ideas such as uniformly sampling $r_1 \ll n$ rows from \math{A} and computing the pseudoinverse of this sampled matrix will not work well for an arbitrary $A$. For example, this idea will fail (with high probability) to return a meaningful approximation for matrices consisting of $n-1$ identical rows and a single row with a nonzero component in the direction perpendicular to that the identical rows; finding that ``outlying'' row is crucial to obtaining a relative-error approximation. This is where the SRHT enters, since it preserves important structures of \math{A}, in particular its rank, by first rotating \math{A} to a random basis and then uniformly sampling rows from the rotated matrix (see~\cite{DMMS07_FastL2_NM10} for more details). More formally, recall that the SVD of \math{A} is \math{U\Sigma V^T} and let $\Pi_1\in\R^{r_1\times n}$ be an \math{\epsilon}-FJLT for \math{U}
(using, for example, the SRHT of Lemma~\ref{thm:theorem7correct} with the appropriate choice for $r_1$).
Then, one could approximate the $\ell_i$'s of Eqn.~(\ref{eqn:pidef-itoA}) by
\begin{equation}\label{eqn:pihat}
\hat{\ell}_i = \TNormS{e_i^T A \left(\Pi_1 A\right)^{\dagger}}   ,
\end{equation}
where we approximated the $n \times d$ matrix $A$ by the $r_1 \times d$ matrix $\Pi_1 A$. Computing \math{A \left(\Pi_1 A\right)^{\dagger}} in this way takes \math{O\left(ndr_1\right)} time, which is not efficient because \math{r_1>d} (from Lemma~\ref{thm:theorem7correct}).

To get around this bottleneck, recall that we only need  the Euclidean norms of the rows of the matrix $A \left(\Pi_1A\right)^{\dagger} \in \mathbb{R}^{n \times r_1}$. Thus, we can further reduce the dimensionality of this matrix by using an \math{\epsilon}-JLT to reduce the dimension \math{r_1=\Omega(d)} to \math{r_2=O(\ln n)}. Specifically, let \math{\Pi_2^T\in\R^{r_2\times r_1}} be an
\math{\epsilon}-JLT for the rows of \math{A \left(\Pi_1A\right)^{\dagger}} (viewed as \math{n} vectors in \math{\R^{r_1}}) and consider the matrix $\Omega = A\left(\Pi_1 A\right)^{\dagger}\Pi_2$.
This $n \times r_2$ matrix $\Omega$ may be viewed as our ``randomized sketch'' of the rows of $AA^\dagger$. Then, we can compute and return
\begin{equation}
\tilde{\ell}_i
   = \TNormS{e_i^TA\left(\Pi_1 A\right)^{\dagger}\Pi_2}  ,
\end{equation}
for each $i\in[n]$, which is essentially what Algorithm~\ref{alg:Leverage_Scores} does. Not surprisingly, the sketch $A\left(\Pi_1 A\right)^{\dagger}\Pi_2$ can be used in other ways: for example, by considering the dot product between two different rows of this randomized sketching matrix (and some additional manipulations) Algorithm~\ref{alg:Off_Diagonal} approximates the large cross-leverage scores of $A$.

\subsection{Approximating all the statistical leverage scores}
\label{sxn:mainalg-diag}

\begin{algorithm}[t]
\begin{framed}

\textbf{Input:} $A \in \mathbb{R}^{n \times d}$ (with SVD
\math{A=U\Sigma V^T}), error parameter $\epsilon \in (0,1/2]$.

\vspace{0.1in}

\textbf{Output:} $\tilde{\ell_i}, i \in [n]$.

\begin{enumerate}

\item Let \math{\Pi_1\in\R^{r_1\times n}} be an
\math{\epsilon}-FJLT for \math{U}, using Lemma~\ref{thm:theorem7correct} with
\mand{r_1 = \Omega\left(\frac{d\ln n}{\epsilon^2}\ln\left(\frac{d\ln n}{\epsilon^2}\right)\right).}
\item Compute \math{\Pi_1 A\in\R^{r_1\times d}} and its SVD,
\math{\Pi_1 A=U_{\Pi_1 A}\Sigma_{\Pi_1 A}V_{\Pi_1 A}^T}. Let
\math{R^{-1}=V_{\Pi_1 A}\Sigma_{\Pi_1 A}^{-1}\in\R^{d\times d}}.
\\
(Alternatively, \math{R} could be computed by a \math{QR} factorization of
\math{\Pi_1 A}.)

\item View the normalized rows of \math{AR^{-1}\in\R^{n\times d}} as \math{n} vectors in \math{\R^{d}}, and construct  \math{\Pi_2\in\R^{d\times r_2}} to be an \math{\epsilon}-JLT for \math{n^2} vectors (the aforementioned \math{n} vectors and their \math{n^2-n} pairwise sums), using Lemma~\ref{lemma:eJLT} with 
    \mand{r_2 = O\left(\epsilon^{-2}\ln n\right).}
    
\item Construct the matrix product \math{\Omega=AR^{-1}\Pi_2}.
\item \textbf{For all} $i \in [n]$ \textbf{compute and return}
$\tilde{\ell}_i=\TNormS{\Omega_{(i)}}$.
\end{enumerate}

\end{framed}
\caption{Approximating the (diagonal) statistical leverage
scores~\math{\ell_i}.}
\label{alg:Leverage_Scores}
\end{algorithm}

Our first main result is Algorithm~\ref{alg:Leverage_Scores}, which takes as input an $n \times d$ matrix $A$ and an error parameter $\epsilon\in(0,1/2]$, and returns as output numbers $\tilde{\ell}_i$, $i\in [n]$. Although the basic idea to approximate \math{\norm{(AA^\dagger)_{(i)}}^2} was described in the previous section, we can improve the efficiency of our approach by avoiding the full sketch of the pseudoinverse. In particular, let \math{\hat A=\Pi_1 A} and let its SVD be \math{\hat A=U_{\hat A}\Sigma_{\hat A}V^T_{\hat A}}. Let \math{R^{-1}=V_{\hat A}\Sigma_{\hat A}^{-1}} and note that \math{R^{-1}\in\R^{d\times d}} is an orthogonalizer for \math{\hat A} since \math{U_{\hat A}=\hat A R^{-1}} is an orthogonal matrix.%
\footnote{This preprocessing is reminiscent of how~\cite{RT08,AMT10} preprocessed the input to provide numerical implementations of the fast relative-error algorithm~\cite{DMMS07_FastL2_NM10} for approximate LS approximation. From this perspective, Algorithm~\ref{alg:Leverage_Scores} can be viewed as specifying a particular basis $Q$, \emph{i.e.}, as choosing $Q$ to be the left singular vectors of $\Pi_1 A$.} 
In addition, note that \math{AR^{-1}} is approximately orthogonal. Thus, we can compute \math{AR^{-1}} and use it as an approximate orthogonal basis for \math{A} and then compute \math{\hat \ell_i} as the squared row-norms of \math{AR^{-1}}. The next lemma states that this is exactly what our main algorithm does; even more, we could get the same estimates by using any ``orthogonalizer'' of \math{\Pi_1 A}.
\begin{lemma}\label{lemma:AnyR}
Let \math{R^{-1}} be such that \math{Q=\Pi_1 A R^{-1}} is an orthogonal matrix with $\rank{Q} = \rank{\Pi_1 A}$. Then, \math{\norm{(AR^{-1})_{(i)}}_2^2=\hat \ell_i}.
\end{lemma}
\begin{proof}
Since \math{\hat A = \Pi_1A} has rank \math{d} (by Lemma~\ref{lemma:FJLT}) and \math{R^{-1}} preserves this rank, \math{R^{-1}} is a \math{d\times d} invertible matrix. Using \math{\hat A=QR} and properties of the pseudoinverse, we get \math{\left(\hat A\right)^\dagger=R^{-1}Q^T}. Thus,
\mand{
\hat \ell_i=\norm{(A\left(\Pi_1 A\right)^\dagger)_{(i)}}_2^2
=
\norm{\left(AR^{-1}Q^T\right)_{(i)}}_2^2
=
\norm{\left(AR^{-1}\right)_{(i)}Q^T}_2^2
=
\norm{\left(AR^{-1}\right)_{(i)}}_2^2.
}
\end{proof}

\noindent
This lemma says that the $\hat{\ell}_i$ of Eqn.~(\ref{eqn:pihat}) can be 
computed with any QR decomposition, rather than with the SVD; but note that 
one would still have to post-multiply by $\Pi_2$, as in 
Algorithm~\ref{alg:Leverage_Scores}, in order to compute ``quickly'' the
approximations of the leverage scores.

\subsection{Approximating the large cross-leverage scores}
\label{sxn:mainalg-offdiag}

By combining Lemmas~\ref{lemma:thm1-a} and~\ref{lemma:thm1-b} (in
Section~\ref{sxn:mainproof-a} below) with the triangle inequality, one
immediately obtains the following lemma.
\begin{lemma}
\label{lemma:CrossAdditive}
Let $\Omega$ be either the sketching matrix constructed by
Algorithm~\ref{alg:Leverage_Scores}, \emph{i.e.}, $\Omega=AR^{-1}\Pi_2$,
or $\Omega=A\left(\Pi_1 A\right)^{\dagger}\Pi_2$ as described in
Section~\ref{sxn:mainalg-basic}.
Then, the pairwise dot-products of the rows of \math{\Omega} are
additive-error approximations to the leverage scores and cross-leverage
scores:
\mand{
\abs{
 \langle U_{(i)},U_{(j)}\rangle-\langle \Omega_{(i)},\Omega_{(j)}\rangle}
\le
\frac{3\epsilon}{1-\epsilon}\norm{U_{(i)}}_2\norm{U_{(j)}}_2.
}
\end{lemma}

\noindent
That is, if one were interested in obtaining an approximation to all the
cross-leverage scores to within additive error (and thus the diagonal
statistical leverage scores to relative-error), then the algorithm which
first computes \math{\Omega} followed by all the pairwise inner products
achieves this in time \math{T(\Omega)+O\left(n^2r_2\right)}, where \math{T(\Omega)} is
the time to compute \math{\Omega} from Section~\ref{sxn:mainalg-diag} and
\math{r_2= O(\epsilon^{-2}\ln n)}.%
\footnote{The exact algorithm which computes a basis first and then the
pairwise inner products requires \math{O(nd^2+n^2d)} time. Thus, by using the
sketch, we can already improve on this running time by a factor of
\math{d/\ln n}.} 
The challenge is to avoid the \math{n^2} computational complexity and
this can be done if one is interested only in the large cross-leverage scores.

Our second main result is provided by Algorithms~\ref{alg:Off_Diagonal} and~\ref{alg:heavy}. 
Algorithm~\ref{alg:Off_Diagonal} takes as input an $n \times d$ matrix $A$, a parameter $\kappa > 1$, and an error parameter $\epsilon\in(0,1/2]$, and returns as output a subset of  $[n] \times [n]$ and estimates \math{\tilde c_{ij}} satisfying Theorem~\ref{thm:main_result-b}. 
The first step of the algorithm is to compute the matrix $\Omega=AR^{-1}\Pi_2$ constructed by Algorithm~\ref{alg:Leverage_Scores}. 
Then, Algorithm~\ref{alg:Off_Diagonal} uses Algorithm~\ref{alg:heavy} as a subroutine to compute ``heavy hitter'' pairs of rows from a~matrix.

\begin{algorithm}[h]
\begin{framed}

\textbf{Input:} $A \in \mathbb{R}^{n \times d}$ and parameters
\math{\kappa>1}, $\epsilon \in (0,1/2]$.

\vspace{0.1in}

\textbf{Output:} The set \math{\cl H} consisting of pairs \math{(i,j)} together with estimates \math{\tilde c_{ij}} satisfying Theorem~\ref{thm:main_result-b}.

\begin{enumerate}
\item
Compute the $n \times r_2$ matrix
$\Omega=AR^{-1}\Pi_2$ from
Algorithm~\ref{alg:Leverage_Scores}.
\item Use Algorithm~\ref{alg:heavy} with inputs \math{\Omega} and
\math{\kappa'=\kappa(1+30 d\epsilon)} to obtain the set
\math{\cl H} containing all the \math{\kappa'}-heavy pairs of
\math{\Omega}.
\item {\bf Return} the pairs in \math{\cl H} as the \math{\kappa}-heavy pairs
of \math{A}.
\end{enumerate}

\end{framed}
\caption{Approximating the large (off-diagonal) cross-leverage scores $c_{ij}$.}
\label{alg:Off_Diagonal}
\end{algorithm}

\begin{algorithm}[t]
\begin{framed}

\textbf{Input:} $X \in \mathbb{R}^{n \times r}$ with rows
\math{x_1,\ldots,x_n} and a parameter \math{\kappa>1}.

\vspace{0.1in}

\textbf{Output:} $\cl H=\{(i,j),\tilde c_{ij}\}$ containing all
heavy (unordered)
pairs. The pair
\math{(i,j),\tilde c_{ij}\in\cl H} if and only if
\math{\tilde c_{ij}^2=\langle x_i,x_j\rangle^2\ge \norm{X^TX}_F^2/\kappa}.

\begin{algorithmic}[1]
\STATE Compute the norms \math{\norm{x_i}_2} and sort the rows according
to norm, so that \math{\norm{x_1}_2\le\cdots\le\norm{x_n}_2}.
\STATE \math{\cl H\gets\{\}}; \math{z_1\gets n}; \math{z_2\gets 1}.
\WHILE{\math{z_2\le z_1}}
\WHILE{\math{\norm{x_{z_1}}_2^2\norm{x_{z_2}}_2^2<\norm{X^TX}_F^2/\kappa}}
\STATE \math{z_2\gets z_2+1}.
\IF{\math{z_2>z_1}}
\STATE {\bf return} \math{\cl H}.
\ENDIF
\ENDWHILE
\FOR{\textbf{each} pair \math{(i,j)} where \math{i=z_1} and \math{j\in\{z_2,z_2+1,\ldots,z_1\}}}
\STATE \math{\tilde c_{ij}^2=\langle x_{i},x_j\rangle^2}. 
\IF{\math{\tilde c_{ij}^2\ge\norm{X^TX}_F^2/\kappa}}
\STATE add \math{(i,j)} and \math{\tilde c_{ij}} to \math{\cl H}.
\ENDIF
\STATE \math{z_1\gets z_1-1}.
\ENDFOR
\ENDWHILE
\STATE {\bf return} \math{\cl H}.
\end{algorithmic}
\end{framed}
\caption{Computing heavy pairs of a matrix.}
\label{alg:heavy}
\end{algorithm}

\section{Proofs of our main theorems}
\label{sxn:mainproof}

\subsection{Sketch of the proof of Theorems ~\ref{thm:main_result}
and~\ref{thm:main_result-b}}
\label{sxn:proofsketch}

We will start by providing a sketch of the proof of Theorems~\ref{thm:main_result} and~\ref{thm:main_result-b}. A detailed proof is provided in the next two subsections. 
In our analysis, we will condition on the events that \math{\Pi_1\in\R^{r_1\times n}} is an \math{\epsilon}-FJLT for \math{U} and \math{\Pi_2\in\R^{r_1\times r_2}} is an \math{\epsilon}-JLT for \math{n^2} points in \math{\R^{r_1}}. 
Note that by setting $\delta = 0.1$ in Lemma~\ref{lemma:eJLT}, both events hold with probability at least $0.8$, which is equal to the success probability of Theorems ~\ref{thm:main_result} and~\ref{thm:main_result-b}. The algorithm estimates \math{\tilde\ell_i=\norm{\tilde u_i}_2^2}, where \math{\tilde u_i=e_i^TA(\Pi_1 A)^\dagger\Pi_2}. First, observe that the sole purpose of \math{\Pi_2} is to improve the running time while preserving pairwise inner products; this is achieved because \math{\Pi_2} is an \math{\epsilon}-JLT for \math{n^2} points. So, the results will follow if \mand{e_i^TA(\Pi_1 A)^\dagger((\Pi_1 A)^\dagger)^T A^Te_j\approx e_i^T U U^Te_j} and \math{(\Pi_1 A)^\dagger} can be computed efficiently.
Since \math{\Pi_1} is an \math{\epsilon}-FJLT for \math{U}, where
\math{A=U\Sigma V^T}, \math{(\Pi_1 A)^\dagger} can be computed
in \math{O(nd\ln r_1+r_1d^2)} time.
By Lemma~\ref{lemma:FJLT},  \math{(\Pi_1 A)^\dagger=V\Sigma^{-1}
(\Pi_1 U)^\dagger}, and so
\mand{
e_i^TA(\Pi_1 A)^\dagger((\Pi_1 A)^\dagger)^T A^Te_j
=
e_i^TU(\Pi_1 U)^\dagger{(\Pi_1 U)^\dagger}^T U^Te_j.}
Since \math{\Pi_1} is an \math{\epsilon}-FJLT for \math{U},
it follows that
\math{(\Pi_1 U)^\dagger{(\Pi_1 U)^\dagger}^T\approx I_d},
 \emph{i.e.},
that \math{\Pi_1 U} is approximately orthogonal. Theorem~\ref{thm:main_result} follows from this basic idea. However, in order to prove Theorem~\ref{thm:main_result-b}, having a sketch which preserves inner products alone is not sufficient. We also need a fast algorithm to identify the large inner products and to relate these to the actual cross-leverage scores. Indeed, it is possible to efficiently find pairs of rows in a general matrix with large inner products. Combining this with the fact that the inner products are preserved, we obtain Theorem~\ref{thm:main_result-b}.

\subsection{Proof of Theorem~\ref{thm:main_result}}
\label{sxn:mainproof-a}

We condition all our analysis on the events that
\math{\Pi_1\in\R^{r_1\times n}} is an \math{\epsilon}-FJLT for \math{U} and
\math{\Pi_2\in\R^{r_1\times r_2}}
is an \math{\epsilon}-JLT for \math{n^2} points in
\math{\R^{r_1}}.
Define
\eqan{\hat u_i&=&e_i^TA(\Pi_1 A)^\dagger, \mbox{\ and}\\
\tilde u_i&=&e_i^TA(\Pi_1 A)^\dagger\Pi_2.}
Then, \math{\hat \ell_i=\norm{\hat u_i}_2^2} and
\math{\tilde \ell_i=\norm{\tilde u_i}_2^2}.
The proof will follow from the following two lemmas.
\begin{lemma}\label{lemma:thm1-a}
For \math{i,j\in[n]},
\mld{\displaystyle
\abs{\langle U_{(i)},U_{(j)}\rangle  - \langle\hat u_i,\hat u_j\rangle}
\leq \frac{\epsilon}{1-\epsilon}\norm{U_{(i)}}_2\norm{U_{(j)}}_2.
\label{eqn:pp1}
}
\end{lemma}
\begin{lemma}\label{lemma:thm1-b}
For \math{i,j\in[n]},
\mld{\displaystyle
\abs{\langle\hat u_i,\hat u_j\rangle - \langle\tilde u_i,\tilde u_j\rangle}
\leq 2\epsilon\norm{\hat u_i}_2\norm{\hat u_j}_2.
\label{eqn:pp2}
}
\end{lemma}

\noindent
Lemma ~\ref{lemma:thm1-a} states that \math{\langle\hat u_i,\hat u_j\rangle}
is an additive error approximation to all the cross-leverage scores ($i \neq j$) and
a relative error approximation for the diagonals ($i =j$). Similarly,
Lemma~\ref{lemma:thm1-b} shows that these cross-leverage scores are preserved
by \math{\Pi_2}. Indeed, with \math{i=j},
from Lemma~\ref{lemma:thm1-a} we have
\math{|\hat \ell_i-\ell_i|\le\frac{\epsilon}{1-\epsilon}\ell_i},
and from Lemma~\ref{lemma:thm1-b} we have
\math{|\hat \ell_i-\tilde \ell_i|\le 2\epsilon\hat \ell_i}.
Using the triangle inequality and $\epsilon \leq 1/2$:
\mand{
\abs{\ell_i - \tilde{\ell}_i}
= \abs{\ell_i - \hat{\ell}_i + \hat{\ell}_i - \tilde{\ell}_i}
\leq \abs{\ell_i - \hat{\ell}_i} + \abs{\hat{\ell}_i - \tilde{\ell}_i}
\leq \left(\frac{\epsilon}{1-\epsilon}+2\epsilon\right) \ell_i
\le 4\epsilon\ell_i.
}
The theorem follows after rescaling \math{\epsilon}.

\paragraph{Proof of Lemma~\ref{lemma:thm1-a}.}
Let \math{A=U\Sigma V^T}.
Using this SVD of \math{A} and
Eqn.~(\ref{lemma:FJLTc})
in Lemma~\ref{lemma:FJLT},
\begin{eqnarray*}
\langle\hat u_i,\hat u_j\rangle
   &=& e_i^T U \Sigma V^T V \Sigma^{-1}  \left(\Pi_1U\right)^{\dagger}
{\left(\Pi_1U\right)^{\dagger}}^T\Sigma^{-1} V^T V \Sigma U^T e_j
   = e_i^T U \left(\Pi_1U\right)^{\dagger}
{\left(\Pi_1U\right)^{\dagger}}^T U^T e_j.
\end{eqnarray*}
By performing standard manipulations, we can now bound
\math{\abs{\langle U_{(i)},U_{(j)}\rangle  - \langle\hat u_i,\hat u_j\rangle}}:
\begin{eqnarray}
\nonumber
\abs{\langle U_{(i)},U_{(j)}\rangle  - \langle\hat u_i,\hat u_j\rangle}
\nonumber
   &=& {e_i^T UU^Te_j -
e_i^T U \left(\Pi_1U\right)^{\dagger}\left(\Pi_1U\right)^{\dagger T}U^Te_j}\\
\nonumber
   &=& {e_i^T U\left(I_d - \left(\Pi_1U\right)^{\dagger}\left(\Pi_1U
\right)^{\dagger T}\right)U^Te_j}\\
\nonumber
   &\leq& \TNorm{I_d - \left(\Pi_1U\right)^{\dagger}\left(\Pi_1U\right)^{\dagger T}}\TNorm{U_{(i)}}\TNorm{U_{(j)}}.
\end{eqnarray}
Let the SVD of \math{\Psi=\Pi_1U} be \math{\Psi=U_{\Psi}\Sigma_{\Psi} V_{\Psi}^T},
where \math{V_\Psi} is a full rotation in \math{d} dimensions (because
\math{\rank{A}=\rank{\Pi_1U}}).
Then, \math{\Psi^\dagger{\Psi^\dagger}^T=V_{\Psi}\Sigma_{\Psi}^{-2}
V_{\Psi}^T}. Thus,
\begin{eqnarray}
\nonumber
\abs{\langle U_{(i)},U_{(j)}\rangle  - \langle\hat u_i,\hat u_j\rangle}
 &\leq& \TNorm{I_d - V_{\Psi}\Sigma_{\Psi}^{-2}V_{\Psi}^T}\TNorm{U_{(i)}}\TNorm{U_{(j)}}\\
\nonumber &=& \TNorm{V_{\Psi}V_{\Psi}^T - V_{\Psi}\Sigma_{\Psi}^{-2}V_{\Psi}^T}\TNorm{U_{(i)}}\TNorm{U_{(j)}}\\
\nonumber &=& \TNorm{I_d - \Sigma_{\Psi}^{-2}}\TNorm{U_{(i)}}\TNorm{U_{(j)}},
\end{eqnarray}
where we used the fact that
$V_{\Psi}V_{\Psi}^T = V_{\Psi}^TV_{\Psi} = I_d$ and the unitary invariance of the spectral norm. Finally, using Eqn.~(\ref{lemma:FJLTb}) of
Lemma~\ref{lemma:FJLT} the result follows.

\paragraph{Proof of Lemma~\ref{lemma:thm1-b}.} 
Since \math{\Pi_2} is an \math{\epsilon}-JLT for \math{n^2} vectors, it preserves the norms of an arbitrary (but fixed) collection of \math{n^2} vectors. Let \math{x_i=\hat u_i/\norm{\hat u_i}_2}.
Consider the following \math{n^2} vectors:
\eqan{
&x_i&\text{for }i\in[n], \mbox{\ and}\\
&x_i+x_j&\text{for }i,j\in[n], i\neq j.
}
By the \math{\epsilon}-JLT property of \math{\Pi_2} and the fact that
\math{\norm{x_i}_2=1},
\eqar{
&1-\epsilon\le \norm{x_i\Pi_2}_2^2\le 1+\epsilon&\text{for }i\in[n], \mbox{\ and}
\label{eq:eJLT-a}\\
&(1-\epsilon)\norm{x_i+x_j}_2^2\le \norm{x_i\Pi_2+x_j\Pi_2}_2^2\le
(1+\epsilon)\norm{x_i+x_j}_2^2&\text{for }i,j\in[n], i\not=j.
\label{eq:eJLT-b}
}
Combining Eqns.~(\ref{eq:eJLT-a}) and~(\ref{eq:eJLT-b}) after expanding the
squares using the identity \math{\norm{a+b}^2=\norm{a}^2+\norm{b}^2+2\langle a,b\rangle},
substituting \math{\norm{x_i}=1}, and after some algebra, we obtain
\mand{
\langle x_i,x_j\rangle-2\epsilon\le
\langle x_i\Pi_2,x_j\Pi_2\rangle\le
\langle x_i,x_j\rangle+2\epsilon.
}
To conclude the proof,
multiply throughout by \math{\norm{\hat u_i}\norm{\hat u_j}} and
use the homogeneity of the inner product, together with the
linearity of \math{\Pi_2}, to obtain:
\mand{
\langle \hat u_i,\hat u_j\rangle-2\epsilon\norm{\hat u_i}\norm{\hat u_j}\le
\langle \hat u_i\Pi_2,\hat u_j\Pi_2\rangle\le
\langle \hat u_i,\hat u_j\rangle+2\epsilon\norm{\hat u_i}\norm{\hat u_j}.
}

\paragraph{Running Times.} By Lemma~\ref{lemma:AnyR}, we can use \math{V_{\Pi_1A}\Sigma_{\Pi_1A}^{-1}} instead of \math{(\Pi_1 A)^\dagger} and obtain the same estimates. Since \math{\Pi_1} is an \math{\epsilon}-FJLT, the product \math{\Pi_1 A} can be computed in \math{O(nd\ln r_1)} while its SVD takes an additional \math{O(r_1d^2)} time to return \math{V_{\Pi_1A}\Sigma_{\Pi_1A}^{-1}\in\R^{d\times d}}. Since \math{\Pi_2\in\R^{d\times r_2}}, we obtain \math{V_{\Pi_1A}\Sigma_{\Pi_1A}^{-1}\Pi_2\in\R^{d\times r_2}} in an additional \math{O(r_2 d^2)} time. Finally, premultiplying by \math{A} takes \math{O(ndr_2)} time, and computing and returning the squared row-norms of \math{\Omega=AV_{\Pi_1A}\Sigma_{\Pi_1A}^{-1}\Pi_2\in\R^{n\times r_2}} takes \math{O\left(nr_2\right)} time. So, the total running time is the sum of all these operations, which is
\mand{O(nd\ln r_1+ndr_2+r_1d^2+r_2d^2).}
For our implementations of the \math{\epsilon}-JLTs and \math{\epsilon}-FJLTs ($\delta = 0.1$), \math{r_1=O\left(\epsilon^{-2}d\left(\ln n\right)\left(\ln\left(\epsilon^{-2}d\ln n\right)\right)\right)} and \math{r_2=O(\epsilon^{-2}\ln n)}. It follows that the asymptotic running time is
\mand{O\left(nd\ln\left(d\epsilon^{-1}\right) + nd\epsilon^{-2}\ln n + d^3\epsilon^{-2}\left(\ln n\right)\left(\ln \left(d\epsilon^{-1}\right)\right)\right).}
To simplify, suppose that \math{d\leq n\leq e^d} and treat \math{\epsilon} as a constant. Then, the asymptotic running time is
\mand{O\left(nd\ln n + d^3 \left(\ln n\right) \left(\ln d\right)\right).}

\subsection{Proof of Theorem~\ref{thm:main_result-b}}
\label{sxn:mainproof-b}

We first  construct an algorithm to estimate the large inner products among the rows of an arbitrary matrix \math{X\in\R^{n\times r}} with \math{n>r}. This general algorithm will be applied to the matrix \math{\Omega=A V_{\Pi_1 A}\Sigma_{\Pi_1 A}^{-1}\Pi_2}. Let $x_1,\ldots,x_n$ denote the rows of $X$; for a given \math{\kappa>1}, the pair \math{(i,j)} is \emph{heavy} if
\mand{\langle x_i,x_j\rangle^2\ge \frac{1}{\kappa}\norm{X^TX}_F^2.}
By the Cauchy-Schwarz inequality, this implies that
\mld{
\norm{x_i}^2_2\norm{x_j}^2_2
\ge \frac{1}{\kappa}\norm{X^TX}_F^2,\label{eq:superlarge}
}
so it suffices to find all the pairs \math{(i,j)} for which Eqn.~(\ref{eq:superlarge}) holds. We will call such pairs \emph{norm-heavy}. Let \math{s} be the number of norm-heavy pairs satisfying Eqn.~(\ref{eq:superlarge}). We first bound the number of
such pairs.
\begin{lemma}
Using the above notation, \math{s\le \kappa r}.
\end{lemma}
\begin{proof}
Observe that
\mand{\sum_{i,j=1}^{n}\norm{x_i}_2^2\norm{x_j}_2^2
=\left(\sum_{i=1}^{n}\norm{x_i}_2^2\right)^2=\norm{X}_F^4
=\left(\sum_{i=1}^r\sigma_i^2\right)^2,
}
where \math{\sigma_1,\ldots,\sigma_r} are the singular values of
\math{X}.
To conclude, by the definition of a heavy pair,
\mand{\sum_{i,j}
\norm{x_i}_2^2\norm{x_j}_2^2\ge \frac{s}{\kappa}\norm{X^TX}_F^2=
\frac{s}{\kappa}\sum_{i=1}^r\sigma_i^4\ge\frac{s}{\kappa r}
\left(\sum_{i=1}^r\sigma_i^2\right)^2,}
where the last inequality follows by Cauchy-Schwarz.
\end{proof}
\noindent Algorithm~\ref{alg:heavy} starts by computing the norms \math{\norm{x_i}_2^2} for all $i \in [n]$ and sorting them (in \math{O\left(nr+n\ln n\right)} time) so that we can assume that \math{\norm{x_1}_2\le\cdots\le\norm{x_n}_2}. Then, we initialize the set of norm-heavy pairs to \math{\cl H=\{\}} and we also initialize two pointers \math{z_1=n} and \math{z_2=1}. The basic loop in the algorithm checks if \math{z_2>z_1} and stops if that is the case. Otherwise, we increment \math{z_2} to the first pair \math{(z_1,z_2)} that is norm-heavy. If none of pairs are norm heavy (\textit{i.e.,} \math{z_2 > z_1} occurs), then we stop and output \math{\cl H}; otherwise, we add \math{(z_1,z_2),(z_1,z_2+1),\ldots,(z_1,z_1)} to \math{\cl H}. This basic loop computes all pairs
\math{(z_1,i)} with \math{i\le z_1} that are norm-heavy. Next, we decrease \math{z_1} by one and if \math{z_1< z_2} we stop and output \math{\cl H}; otherwise, we repeat the basic loop. Note that in the basic loop \math{z_2} is always \emph{incremented}. This occurs whenever the pair \math{(z_1,z_2)} is not norm-heavy. Since \math{z_2} can be incremented at most \math{n} times, the number of times we check whether a pair is norm-heavy and fail is at most \math{n}. Every successful check results in the addition of at least one norm-heavy pair into \math{\cl H} and thus the number of times we check if a pair is norm heavy (a constant-time operation) is at most \math{n+s}. The number of pair additions into \math{\cl H} is exactly \math{s} and thus the total running time is \math{O(nr+n\ln n+s)}. Finally, we must check each norm-heavy pair to verify whether or not it is actually heavy by computing \math{s} inner products vectors in $\mathbb{R}^r$; this can be done in \math{O(sr)} time.
Using \math{s\le \kappa r} we get the following lemma.
\begin{lemma}
\label{lemma:heavyX}
Algorithm~\ref{alg:heavy} returns \math{\cl H} including all the heavy pairs of \math{X} in \math{O(nr+\kappa r^2+n\ln n)} time.
\end{lemma}
\noindent To complete the proof, we apply Algorithm~\ref{alg:heavy} with \math{\Omega=AV_{\Pi_1 A}\Sigma_{\Pi_1 A}^{-1}\Pi_2\in\R^{n\times r_2}}, where \math{r_2=O(\epsilon^{-2}\ln n)}. Let \math{\tilde u_1,\ldots,\tilde u_n} denote the rows of \math{\Omega} and recall that \math{A=U\Sigma V^T}. Let \math{u_1,\ldots,u_n} denote the rows of \math{U}; then, from Lemma~\ref{lemma:CrossAdditive},
\mld{
\langle  u_i, u_j\rangle-
\frac{3\epsilon}{1-\epsilon}\norm{u_i}\norm{u_j}
\le
\langle \tilde u_i,\tilde u_j\rangle
\le
\langle u_i, u_j\rangle+
\frac{3\epsilon}{1-\epsilon}\norm{u_i}\norm{u_j}.
\label{eq:approx}
}
Given \math{\epsilon,\kappa}, assume that for the pair of vectors $u_i$ and $u_j$
\mand{
\langle u_i,u_j\rangle^2
\ge
\frac{1}{\kappa}
\norm{U^TU}_F^2+
12\epsilon\norm{u_i}^2\norm{u_j}^2
=
\frac{d}{\kappa}+12\epsilon\norm{u_i}^2\norm{u_j}^2,
}
where the last equality follows from \math{\norm{U^TU}_F^2=\norm{I_d}_F^2 =d}. By Eqn.~(\ref{eq:approx}), after squaring and using \math{\epsilon<0.5},
\eqar{
\langle u_i,u_j\rangle^2-
12\epsilon\norm{u_i}^2\epsilon\norm{u_j}^2
\le
\langle \tilde u_i,\tilde u_j\rangle^2
\le
\langle u_i,u_j\rangle^2
+30\epsilon\norm{u_i}^2\norm{u_j}^2  .
\label{eq:approx1}
}
Thus, \math{\langle \tilde u_i,\tilde u_j\rangle^2\ge d/\kappa} and summing Eqn.~(\ref{eq:approx1}) over all \math{i,j} we get \math{\norm{\Omega^T\Omega}_F^2\le d+30\epsilon d^2}, or, equivalently,
\mand{
d\ge\frac{\norm{\Omega^T\Omega}_F^2}{1+30d\epsilon}.
}
We conclude that
\mld{
\langle u_i,u_j\rangle^2
\ge
\frac{d}{\kappa}
+
12\epsilon\norm{u_i}^2\norm{u_j}^2
\implies
\langle \tilde u_i,\tilde u_j\rangle^2
\ge
\frac{d}{\kappa}
\ge
\frac{\norm{\Omega^T\Omega}_F^2}{\kappa(1+30d\epsilon)}.
\label{eq:bound-d}
}
By construction, Algorithm~\ref{alg:heavy} is invoked with \math{\kappa'=\kappa\norm{\Omega^T\Omega}_F^2/d} and thus it finds all pairs with \math{\langle \tilde u_i,\tilde u_j\rangle^2\ge \norm{\Omega^T\Omega}_F^2/\kappa'=d/\kappa}. This set contains all pairs for which \mand{\langle u_i,u_j\rangle^2 \ge \frac{d}{\kappa} + 12\epsilon\norm{u_i}^2\norm{u_j}^2.} Further, since every pair
returned satisfies \math{\langle \tilde u_i,\tilde u_j\rangle^2\ge d/\kappa}, by Eqn.~(\ref{eq:approx1}), \math{c_{ij}\ge d/\kappa-30\epsilon\ell_i\ell_j}. This proves the first claim of the Theorem; the second claim follows analogously from Eqn.~(\ref{eq:approx1}).

Using Lemma~\ref{lemma:heavyX}, the running time of our approach is \math{O\left(nr_2+\kappa'r_2^2+n\ln n\right)}. Since \math{r_2=O\left(\epsilon^{-2}\ln n\right)}, and, by
Eqn.~(\ref{eq:bound-d}), \math{\kappa'=\kappa\norm{\Omega^T\Omega}_F^2/d\le \kappa(1+30d\epsilon)}, the overall running time is \math{O\left(\epsilon^{-2}n\ln n+\epsilon^{-3}\kappa d\ln^2n\right)}. 

\section{Extending our algorithm to general matrices} \label{sxn:extensions}

In this section, we will describe an important extension of our main result, namely the computation of the statistical leverage scores relative to the best rank-$k$ approximation to a general matrix $A$. More specifically, we consider the estimation of leverage scores for the case of general ``fat'' matrices, namely input matrices $A\in\mathbb{R}^{n \times d}$, where both $n$ and $d$ are large, \emph{e.g.}, when $d=n$ or $d = \Theta(n)$. Clearly, the leverage scores of any full rank $n\times n$ matrix are exactly uniform. The problem becomes interesting if one specifies a rank parameter $k \ll \min\{n,d\}$. This may arise when the numerical rank of $A$ is small (\emph{e.g.}, in some scientific computing applications, more than $99\%$ of the spectral norm of $A$ may be captured by some $k \ll \min\{n,d\}$ directions), or, more generally, when one is interested in some low rank approximation to $A$ (\emph{e.g.}, in some data analysis applications, a reasonable fraction or even the majority of the Frobenius norm of $A$ may be captured by some $k \ll \min\{n,d\}$ directions, where $k$ is determined by some exogenously-specified model selection criterion). Thus, assume that in addition to a general $n \times d$ matrix $A$, a rank parameter $k< \min\{n,d\}$ is specified. In this case, we wish to obtain the statistical leverage scores $\ell_i = \norm{(U_k)_{(i)}}_2^2$ for $A_k=U_k\Sigma_kV_k^T$, the best rank-$k$ approximation to $A$. Equivalently, we seek the normalized leverage scores
\begin{equation}\label{eqn:pp10}
p_i=\frac{\ell_i}{k}.
\end{equation}
Note that $\sum_{i=1}^n p_i = 1$ since $\sum_{i=1}^n \ell_i = \norm{U_k}_F^2 = k$.

Unfortunately, as stated, this is an ill-posed problem. Indeed, consider the degenerate case when $A=I_n$ (\emph{i.e.}, the $n \times n$ identity matrix). In this case, $U_k$ is not unique and the leverage scores are not well-defined. Moreover, for the obvious $\choose{n}{k}$ equivalent choices for $U_k$, the leverage scores defined according to any one of these choices do not provide a relative error approximation to the leverage scores defined according to any other choices. More generally, removing this trivial degeneracy does not help. Consider the matrix
$$
A=\begin{pmatrix}I_k&0\\0&(1-\gamma)I_{n-k} \end{pmatrix} \in \mathbb{R}^{n \times n}.
$$
In this example, the leverage scores for $A_k$ are well defined. However, as $\gamma\rightarrow0$, it is not possible to distinguish between the top-$k$ singular space and its complement. This example suggests that it should be possible to obtain some result conditioning on the spectral gap at the $k^{th}$ singular value. For example, one might assume that $\sigma^2_k-\sigma^2_{k+1}\ge \gamma>0$, in which case the parameter $\gamma$ would play an important role in the ability to solve this problem.
Any algorithm which cannot distinguish the singular values with an error less than $\gamma$ will confuse the $k$-th and $(k+1)$-th singular vectors and consequently will fail to get an accurate approximation to the leverage scores for $A_k$.

In the following, we take a more natural approach which leads to a clean problem formulation. To do so, recall that the leverage scores and the related normalized leverage scores of Eqn.~(\ref{eqn:pp10}) are used to approximate the matrix in some way, \emph{e.g.}, we might be seeking a low-rank approximation to the matrix with respect to the spectral~\cite{DMM08_CURtheory_JRNL} or the Frobenius~\cite{BMD09_CSSP_SODA} norm, or we might be seeking useful features or data points in downstream data analysis applications~\cite{Paschou07b,CUR_PNAS}, or we might be seeking to develop high-quality numerical implementations of low-rank matrix
approximation algorithms~\cite{HMT09_SIREV}, etc. In all these cases, we only care that the estimated leverage scores are a good approximation to the leverage scores of some ``good'' low-rank approximation to $A$. The following definition captures the notion of a set of rank-$k$ matrices that are good approximations to $A$.
\begin{definition}\label{def:se}
Given $A \in \mathbb{R}^{n \times d}$ and a rank parameter $k \ll \min\left\{n,d\right\}$, let $A_k$ be the best rank-$k$ approximation to $A$. Define the set $\cl S_{\epsilon}$ of rank-$k$ matrices that are good approximations to $A$ as follows (for $\xi = 2,F$):
\begin{equation}\label{eqn:def-of-S}
\cl S_{\epsilon} = \left\{ X\in\mathbb{R}^{n\times d}:\ \rank{X}=k\ \mbox{\ and \ }\norm{A-X}_{\xi}\le(1+\epsilon)\norm{A- A_k}_{\xi}\right\} .
\end{equation}
\end{definition}
\noindent We are now ready to define our approximations to the normalized leverage scores of any matrix $A \in \mathbb{R}^{n \times d}$ given a rank parameter $k \ll \min\left\{n,d\right\}$. Instead of seeking to approximate the $p_i$ of Eqn.~(\ref{eqn:pp10}) (a problem that is ill-posed as discussed above), we will be satisfied if we can approximate the normalized leverage scores of some matrix $X \in \cl S_{\epsilon}$. This is an interesting relaxation of the task at hand: all matrices $X$ that are sufficiently close to $A_k$ are essentially equivalent, since they can be used instead of $A_k$ in applications.
\begin{definition}\label{def:applev}
Given $A \in \mathbb{R}^{n \times d}$ and a rank parameter $k \ll \min\left\{n,d\right\}$, let $\cl S_{\epsilon}$ be the set of matrices of Definition~\ref{def:se}. We call the numbers $\hat{p}_i$ (for all $i \in [n]$) $\beta$-approximations to the normalized leverage scores of $A_k$ (the best rank-$k$ approximation to $A$) if, for some matrix $X \in \cl S_{\epsilon}$,
$$\hat{p}_i \geq \frac{\beta\norm{({U_X})_{(i)}}_2^2}{k} \qquad \mbox{and} \qquad \sum_{i=1}^n \hat{p}_i = 1.$$
Here $U_X \in \mathbb{R}^{n \times k}$ is the matrix of the left singular vectors of $X$.
\end{definition}
\noindent Thus, we will seek algorithms whose output is a set of numbers, with the requirement that those numbers are good approximations to the normalized leverage scores of some matrix $X\in\cl S_\epsilon$ (instead of $A_k$). This removes the ill-posedness of the original problem. Next, we will give two examples of algorithms that compute such $\beta$-approximations to the normalized leverage scores of a general matrix $A$ with a rank parameter $k$ for two popular norms, the spectral norm and the Frobenius norm.%
\footnote{Note that we will not compute $S_\epsilon$, but our algorithms 
will compute a matrix in that set.  Moreover, that matrix can be used for 
high-quality low-rank matrix approximation.  See the comments in 
Section~\ref{sxn:intro-empirical} for more details.}

\subsection{Leverage Scores for Spectral Norm Approximators}
Algorithm~\ref{alg:GeneralSpectral} approximates the statistical leverage scores of a general matrix $A$ with rank parameter $k$ in the spectral norm case. It takes as inputs a matrix $A \in \math{R}^{n \times d}$ with $\rank{A}=\rho$ and a rank parameter $k \ll \rho$, and outputs a set of numbers $\hat{p}_i$ for all $i \in [n]$, namely our approximations to the normalized leverage scores of $A$ with rank parameter $k.$

\begin{algorithm}[t]
\begin{framed}

\textbf{Input:} $A \in \mathbb{R}^{n \times d}$ with $\rank{A}=\rho$ and a rank parameter $k \ll \rho$

\vspace{0.1in}

\textbf{Output:} $\hat{p_i}, i \in [n]$
\begin{enumerate}
\item Construct $\Pi\in\mathbb{R}^{d\times 2k}$ with entries drawn in i.i.d. trials from the normal distribution $\cl N(0,1)$.
\item Compute $B=\left(A A^T\right)^q A \Pi \in \mathbb{R}^{n \times 2k}$, with $q$ as in Eqn.~(\ref{eqn:defq}).
\item Approximately compute the statistical leverage scores of the ``tall'' matrix $B$ by calling Algorithm~\ref{alg:Leverage_Scores} with inputs $B$ and $\epsilon$; let $\hat{\ell}_i$ (for all $i \in [n]$) be the outputs of Algorithm~\ref{alg:Leverage_Scores}.
\item Return
\begin{equation}\label{eqn:ourapprox}
\hat{p}_i = \frac{\hat{\ell}_i}{\sum_{j=1}^n \hat{\ell}_j}
\end{equation}
for all $i \in [n]$.
\end{enumerate}
\end{framed}
\caption{Approximating the statistical leverage
scores of a general matrix $A$ (spectral norm case).}
\label{alg:GeneralSpectral}
\end{algorithm}

The next lemma argues that there exists a matrix $X \in \mathbb{R}^{n \times d}$ of rank $k$ that is sufficiently close to $A$ (in particular, it is a member of $\cl S_{\epsilon}$ with constant probability) and, additionally, can be written as
$X = BY,$
where $Y\in \mathbb{R}^{2k \times d}$ is a matrix of rank $k$. 
A version of this lemma was essentially proven in~\cite{HMT09_SIREV}, but see 
also~\cite{RST09} for computational details; we will use the version of the lemma that appeared in~\cite{BDM11_TR}.  (See also the conference version~\cite{BDM11}, but in the remainder we refer to the technical report version~\cite{BDM11_TR} for consistency of numbering.)  Note that for our purposes in this section, the computation of $Y$ is not relevant and we defer the reader to~\cite{HMT09_SIREV, BDM11_TR} for details.
\begin{lemma}[Spectral Sketch]
\label{lem:specsketch}
Given $A\in\mathbb{R}^{n\times d}$ of rank $\rho$, a rank parameter $k$ such that $2\leq k < \rho$, and an error parameter $\epsilon$ such that $0 < \epsilon < 1$, let $\Pi\in\mathbb{R}^{d\times 2k}$ be a standard Gaussian matrix (with entries selected in i.i.d. trials from $\cl N(0,1)$). If $B=\left(A A^T\right)^q A \Pi$, where
\begin{equation}\label{eqn:defq}
q \geq \left\lceil \frac{\ln\left(1+\sqrt{\frac{k}{k-1}}+e\sqrt{\frac{2}{k}}\sqrt{\min\left\{n,d\right\}-k}\right)}{2\ln \left(1+\epsilon/10\right)-1/2} \right\rceil,
\end{equation}
then there exists a matrix $X \in \mathbb{R}^{n \times d}$ of rank $k$ satisfying $X = BY$ (with $Y\in \mathbb{R}^{2k \times d}$) such that
$$
\Exp\left[\norm{A-X}_2\right] \le \left(1+\frac{\epsilon}{10}\right)\norm{A-A_k}_2.
$$
The matrix $B$ can be computed in $O\left(ndkq\right)$ time.
\end{lemma}
\noindent This version of the above lemma is proven in~\cite{BDM11_TR}.%
\footnote{More specifically, the proof may be found in Lemma~32 and in particular in Eqn.~(14) in Section A.2; note that for our purposes here we replaced $\epsilon/\sqrt{2}$ by $\epsilon/10$ after adjusting $q$ accordingly.}
Now, since $X$ has rank $k$, it follows that $\norm{A-X}_2 \geq \norm{A-A_k}_2$ and thus we can consider the non-negative random variable $\norm{A-X}_2-\norm{A-A_k}_2$ and apply Markov's inequality to get that
$$\norm{A-X}_2 - \norm{A-A_k}_2 \le \epsilon\norm{A-A_k}_2$$
holds with probability at least $0.9$. Thus, $X \in \cl S_{\epsilon}$ with probability at least $0.9$.

The next step of the proposed algorithm is to approximately compute the leverage scores of $B \in \mathbb{R}^{n \times 2k}$ via Algorithm~\ref{alg:Leverage_Scores}. Under the assumptions of Theorem~\ref{thm:main_result}, this step runs in $O\left(nk\epsilon^{-2}\ln n\right)$ time. Let $U_X \in \mathbb{R}^{n \times k}$ be the matrix containing the left singular vectors of the matrix $X$ of Lemma~\ref{lem:specsketch}. Then, since $X = BY$ by Lemma~\ref{lem:specsketch}, it follows that
$$U_B = \left[U_X\ \ U_{R}\right]$$
is a basis for the subspace spanned by the columns of $B$. Here $U_R\in \mathbb{R}^{n \times k}$ is an orthogonal matrix whose columns are perpendicular to the columns of $U_X$. Now consider the approximate leverage scores $\hat{\ell}_i$ computed by Algorithm~\ref{alg:Leverage_Scores} and note that (by Theorem~\ref{thm:main_result}),
$$\abs{\hat{\ell}_i - \norm{\left(U_B\right)_{(i)}}_2^2} \leq \epsilon \norm{\left(U_B\right)_{(i)}}_2^2$$
holds with probability at least $0.8$ for all $i \in [n]$. It follows that
$$\sum_{j=1}^n \hat{\ell}_j \leq \left(1+\epsilon\right)\sum_{j=1}^n \norm{\left(U_B\right)_{(j)}}_2^2=
\left(1+\epsilon\right)\sum_{j=1}^n \norm{U_B}_F^2 = 2\left(1+\epsilon\right)k.
$$
Finally,
\begin{eqnarray*}
\hat{p}_i = \frac{\hat{\ell}_i}{\sum_{j=1}^n \hat{\ell}_j} &\geq& \left(1-\epsilon\right) \frac{\norm{\left(U_B\right)_{(i)}}_2^2}{\sum_{j=1}^n \hat{\ell}_j}\\
&\geq& \left(1-\epsilon\right) \frac{\norm{\left(U_X\right)_{(i)}}_2^2 + \norm{\left(U_R\right)_{(i)}}_2^2}{\sum_{j=1}^n \hat{\ell}_j}\\
&\geq& \frac{1-\epsilon}{2} \frac{\norm{\left(U_X\right)_{(i)}}_2^2}{\sum_{j=1}^n \hat{\ell}_j}\\
&\geq& \frac{1-\epsilon}{2\left(1+\epsilon\right)} \frac{\norm{\left(U_X\right)_{(i)}}_2^2}{k}.
\end{eqnarray*}
Clearly, $\norm{\left(U_X\right)_{(i)}}_2^2/k$ are the normalized leverage scores of the matrix $X$. Recall that $X \in \cl S_{\epsilon}$ with probability at least $0.9$ and use Definition~\ref{def:applev} to conclude that the scores $\hat{p}_i$ of Eqn.~(\ref{eqn:ourapprox}) are $\left(\frac{1-\epsilon}{2\left(1+\epsilon\right)}\right)$-approximations to the normalized leverage scores of $A$ with rank parameter $k$. The following Theorem summarizes the above discussion:
\begin{theorem}\label{theorem:LevFatSpectral}
Given $A\in\mathbb{R}^{n\times d}$, a rank parameter $k$, and an accuracy parameter $\epsilon$, Algorithm~\ref{alg:GeneralSpectral} computes a set of normalized leverage scores $\hat{p}_i$ that are $\left(\frac{1-\epsilon}{2\left(1+\epsilon\right)}\right)$-approximations to the normalized leverage scores of $A$ with rank parameter $k$ with probability at least $0.7$. The proposed algorithm runs in $$O\left(ndk\frac{\ln\left(\min\{n,d\}\right)}{\ln\left(1+\epsilon\right)}+nk\epsilon^{-2}\ln n\right)$$ time.
\end{theorem}

\subsection{Leverage Scores for Frobenius Norm Approximators.}
Algorithm~\ref{alg:GeneralFrobenius} approximates the statistical leverage scores of a general matrix $A$ with rank parameter $k$ in the Frobenius norm case. It takes as inputs a matrix $A \in \math{R}^{n \times d}$ with $\rank{A}=\rho$ and a rank parameter $k \ll \rho$, and outputs a set of numbers $\hat{p}_i$ for all $i \in [n]$, namely our approximations to the normalized leverage scores of $A$ with rank parameter $k.$
%
\begin{algorithm}[t]
\begin{framed}

\textbf{Input:} $A \in \mathbb{R}^{n \times d}$ with $\rank{A}=\rho$ and a rank parameter $k \ll \rho$

\vspace{0.1in}

\textbf{Output:} $\hat{p_i}, i \in [n]$
\begin{enumerate}
\item Let $r$ be as in Eqn.~(\ref{eqn:defr}) and construct $\Pi\in\mathbb{R}^{d\times r}$ whose entries are drawn in i.i.d. trials from the normal distribution $\cl N(0,1)$.
\item Compute $B=A \Pi \in \mathbb{R}^{n \times r}$.
\item Compute a matrix $Q \in \mathbb{R}^{n \times r}$ whose columns form an orthonormal basis for the column space of $B$.
\item Compute the matrix $Q^TA \in \mathbb{R}^{r \times d}$ and its left singular vectors $U_{Q^TA} \in \mathbb{R}^{r \times d}$.
\item Let $U_{Q^TA,k} \in \mathbb{R}^{r \times k}$ denote the top $k$ left singular vectors of the matrix $Q^TA$ (the first $k$ columns of $U_{Q^TA}$) and compute, for all $i \in [n]$,
\begin{equation}\label{eqn:ourapprox2}
\hat{\ell}_i = \norm{\left(QU_{Q^TA,k}\right)_{(i)}}_2^2.
\end{equation}
\item Return $\hat{p}_i = \hat{\ell}_i/k$ for all $i \in [n]$.
\end{enumerate}
\end{framed}
\caption{Approximating the statistical leverage
scores of a general matrix $A$ (Frobenius norm case).}
\label{alg:GeneralFrobenius}
\end{algorithm}
%
It is worth noting that $\sum_{i=1}^n \hat{\ell}_i = \norm{QU_{Q^TA,k}}_F^2 = \norm{U_{Q^TA,k}}_F^2 = k$ and thus the $\hat{p}_i$ sum up to one. The next lemma argues that there exists a matrix $X \in \mathbb{R}^{n \times d}$ of rank $k$ that is sufficiently close to $A$ (in particular, it is a member of $\cl S_{\epsilon}$ with constant probability). Unlike the previous section (the spectral norm case), we will now be able to provide a closed-form formula for this matrix $X$ and, more importantly, the normalized leverage scores of $X$ will be \textit{exactly equal} to the $\hat{p}_i$ returned by our algorithm. Thus, in the parlance of Definition~\ref{def:applev}, we will get a 1-approximation to the normalized leverage scores of $A$ with rank parameter $k$.

\begin{lemma}[Frobenius Sketch]
\label{lem:Frobeniussketch}
Given $A\in\mathbb{R}^{n\times d}$ of rank $\rho$, a rank parameter $k$ such that $2\leq k < \rho$, and an error parameter $\epsilon$ such that $0 < \epsilon < 1$, let $\Pi\in\mathbb{R}^{d\times r}$ be a standard Gaussian matrix (with entries selected in i.i.d. trials from $\cl N(0,1)$) with
\begin{equation}\label{eqn:defr}
r \geq k+\left\lceil \frac{10k}{\epsilon} + 1 \right\rceil.
\end{equation}
Let $B=A \Pi$ and let $X$ be as in Eqn.~(\ref{eqn:defX}). Then,
$$
\Exp\left[\norm{A-X}_F^2\right] \le \left(1+\frac{\epsilon}{10}\right)\norm{A-A_k}_F^2.
$$
The matrix $B$ can be computed in $O\left(ndk\epsilon^{-1}\right)$ time.
\end{lemma}
\noindent Let
\begin{equation}\label{eqn:defX}
X = Q\left(Q^TA\right)_k \in \mathbb{R}^{n \times d},
\end{equation}
where $\left(Q^TA\right)_k$ is the best rank-$k$ approximation to the matrix $Q^TA$; from standard linear algebra, $\left(Q^TA\right)_k = U_{Q^TA,k} U_{Q^TA,k}^T Q^TA$. Then, the above lemma is proven in~\cite{BDM11_TR}.%
\footnote{More specifically, the proof may be found in Lemma~33 in Section A.3; note that for our purposes here we set $p = \left\lceil \frac{10k}{\epsilon} + 1 \right\rceil$.} 
Now, since $X$ has rank $k$, it follows that $\norm{A-X}_F^2 \geq \norm{A-A_k}_F^2$ and thus we can consider the non-negative random variable $\norm{A-X}_F^2-\norm{A-A_k}_F^2$ and apply Markov's inequality to get that
$$\norm{A-X}_F^2 - \norm{A-A_k}_F^2 \le \epsilon\norm{A-A_k}_F^2$$
holds with probability at least $0.9$. Rearranging terms and taking square roots of both sides implies that
$$\norm{A-X}_F \leq \sqrt{1+\epsilon}\norm{A-A_k}_F \leq \left(1+\epsilon\right) \norm{A-A_k}_F.$$
Thus, $X \in \cl S_{\epsilon}$ with probability at least $0.9$. To conclude our proof, recall that $Q$ is an orthonormal basis for the columns of $B$.
From Eqn.~(\ref{eqn:defX}),
$$X = Q\left(Q^TA\right)_k = QU_{Q^TA,k} U_{Q^TA,k}^T Q^TA = QU_{Q^TA,k} \Sigma_{Q^TA,k} V_{Q^TA,k}^T.$$
In the above, $\Sigma_{Q^TA,k}\in \mathbb{R}^{k \times k}$ is the diagonal matrix containing the top $k$ singular values of $Q^TA$ and $V_{Q^TA,k}^T \in \mathbb{R}^{k \times d}$ is the matrix whose rows are the top $k$ right singular vectors of $Q^TA$. Thus, the left singular vectors of the matrix $X$ are exactly equal to the columns of the orthogonal matrix $QU_{Q^TA,k}$; it now follows that the $\hat{\ell}_i$ of Eqn.~(\ref{eqn:ourapprox2}) are the leverage scores of the matrix $X$ and, finally, that the $\hat{p}_i$ returned by the proposed algorithm are the normalized leverage scores of the matrix $X$.

We briefly discuss the running time of the proposed algorithm. First, we can compute $B$ in $O(ndr)$ time. Then, the computation of $Q$ takes $O(nr^2)$ time. The computation of $Q^TA$ takes $O(ndr)$ time and the computation of $U_{Q^TA}$ takes $O(dr^2)$ time. Thus, the total time is equal to $O\left(ndr + (n+d)r^2 \right)$. The following Theorem summarizes the above discussion.
\begin{theorem}\label{theorem:LevFatFrobenius}
Given $A\in\mathbb{R}^{n\times d}$, a rank parameter $k$, and an accuracy parameter $\epsilon$, Algorithm~\ref{alg:GeneralFrobenius} computes a set of normalized leverage scores $\hat{p}_i$ that are $1$-approximations to the normalized leverage scores of $A$ with rank parameter $k$ with probability at least $0.7$. The proposed algorithm runs in $O\left(ndk\epsilon^{-1} + (n+d)k^2\epsilon^{-2}\right)$ time.
\end{theorem} 

\section{Discussion}\label{sxn:discussion}

We will conclude with a discussion of our main results in a broader context: 
understanding the relationship between our main algorithm and a related 
estimator for the statistical leverage scores; applying our main algorithm 
to solve under-constrained least squares problems; and implementing variants 
of the basic algorithm in streaming environments.

\subsection{A related estimator for the leverage scores}

Magdon-Ismail in~\cite{Malik10_TR} presented the following algorithm to estimate the statistical leverage scores: given as input an $n \times d$ matrix $A$, with $n \gg d$, the algorithm proceeds as follows.
\begin{itemize}
\item
Compute $\Pi A$, where the
$O\left( \frac{n\ln d }{\ln^2 n } \right) \times n $ matrix $\Pi$ is a
SRHT or another FJLT.
\item
Compute $X = (\Pi A)^{\dagger} \Pi$.
\item
For $t=1,\ldots,n$, compute the estimate $\tilde{w}_t = A_{(t)}^TX^{(t)}$
and set $w_t = \max\left\{\frac{d\ln^2 n }{4n},\tilde{w}_t\right\}$.
\item
Return the quantities $\tilde{p}_i = w_i / \sum_{i'=1}^{n} w_{i'} $, for
$i \in [n]$.
\end{itemize}
\noindent \cite{Malik10_TR} argued that
the output $\tilde{p}_i$ achieves an $O(\ln^2n)$ approximation to all of the
(normalized) statistical leverage scores of $A$ in roughly $O(nd^2/\ln n)$
time. (To our knowledge, prior to our work here, this is the only known estimator that obtains
any nontrivial provable approximation to the leverage scores of a matrix in $o(nd^2)$ time.)
To see the relationship between this estimator and our main result, recall
that
$$
\ell_i = e_i^TUU^Te_i = e_i^TAA^\dagger e_i = x_i^Ty_i ,
$$
where the vector $x_i^T=e_i^TA$ is cheap to compute and the vector $y_i = A^\dagger e_i$ is expensive to compute. The above algorithm effectively approximates $y_i = A^\dagger e_i$ via a random
projection as $\tilde{y}_i=(\Pi A)^\dagger \Pi e_i$, where $\Pi$ is a SRHT or another FJLT. Since the estimates $x_i^T\tilde{y}_i$ are not necessarily positive, a
truncation at the negative tail, followed by a renormalization step, must be performed in order to arrive at the final estimator returned by the algorithm. This truncation-renormalization step has the effect of inflating the estimates of the small leverage scores by an $O(\ln^2 n )$ factor. By way of comparison, Algorithm~\ref{alg:Leverage_Scores} essentially computes a sketch of $AA^\dagger$ of the form $ A(\Pi A)^\dagger\Pi^{T} $ that maintains positivity for each of the row norm~estimates.

Although both Algorithm~\ref{alg:Leverage_Scores} and the algorithm of this subsection estimate $AA^\dagger$ by a matrix of the form $A(\Pi A)^\dagger\Pi^{T}$, there are notable differences. The algorithm of this subsection does not actually compute or approximate $AA^T$ directly; instead, it separates the matrix into two parts and computes the dot product between $e_i^TA$ and $(\Pi A)^\dagger\Pi e_i$. 
Positivity is sacrificed and this leads to some complications in the estimator; however, the truncation step is interesting, since, despite the fact that the estimates are ``biased'' (in a manner somewhat akin to what is obtained with ``thresholding'' or ``regularization'' procedures), we still obtain provable approximation guarantees. The algorithm of this subsection is simpler (since it uses an application of only one random projection), albeit at the cost of weaker theoretical guarantees and a worse running time than our main algorithm. A direction of considerable practical interest is to evaluate empirically the performance of these two estimators, either for estimating all the leverage scores or (more interestingly) for estimating the largest leverage
scores for data matrices for which the leverage scores are quite nonuniform.

\subsection{An application to under-constrained least-squares problems}

Consider the following under-constrained least-squares problem:
\begin{equation}\label{eqn:underLS}
\min_{x \in \mathbb{R}^d}\TNorm{Ax - b},
\end{equation}
where $A \in \mathbb{R}^{n \times d}$ has much fewer rows than columns, \emph{i.e.}, $n \ll d$. It is well-known that we can solve this problem exactly in $O(n^2d)$ time and that the minimal $\ell_2$-norm solution is given by
$ x_{opt} = A^{\dagger} b. $
For simplicity, let's assume that the input matrix  $A$ has full rank (\emph{i.e.}, $\rank{A}=n$) and thus $\TNorm{Ax_{opt} - b} = 0$.

In this section, we will argue that Algorithm~\ref{alg:alg_sample_fast} computes a simple, accurate estimator $\tilde{x}_{opt}$ for $x_{opt}$. In words, Algorithm~\ref{alg:alg_sample_fast} samples a small number of columns from $A$ (note that the columns of $A$ correspond to variables in our under-constrained problem) and uses the sampled columns to compute $\tilde{x}_{opt}$. However, in order to determine which columns will be included in the sample, the algorithm will make use of the statistical leverage scores of the matrix $A^T$; more specifically, columns (and thus variables) will be chosen with probability proportional to the corresponding statistical leverage score. We will state Algorithm~\ref{alg:alg_sample_fast} assuming that these probabilities are parts of the input; the following theorem is our main quality-of-approximation result for Algorithm~\ref{alg:alg_sample_fast}.
\begin{theorem}\label{thm:underLS}
Let $A \in \mathbb{R}^{n \times d}$ be a full-rank matrix with $n \ll d$; let $\epsilon \in (0,0.5]$ be an accuracy parameter; let $\delta \in (0,1)$ be a failure probability; and let $x_{opt} = A^{\dagger} b$ be the minimal $\ell_2$-norm solution to the least-squares problem of Eqn.~(\ref{eqn:underLS}). Let $p_i \geq 0$, $i \in [d]$, be a set of probabilities satisfying $\sum_{i=1}^d p_i = 1$ and
\begin{equation}\label{eqn:beta_underLS}
p_i \geq \frac{\beta\TNormS{V_{(i)}}}{n}
\end{equation}
for some constant $\beta \in (0,1]$. (Here $V \in \mathbb{R}^{d \times n}$ is the matrix of the right singular vectors of $A$.) If $\tilde{x}_{opt}$ is computed via Algorithm~\ref{alg:alg_sample_fast} then, with probability at least $1-\delta$,
$$\TNorm{x_{opt}-\tilde{x}_{opt}}\leq 2\epsilon \TNorm{x_{opt}}.$$
Algorithm~\ref{alg:alg_sample_fast} runs in $O\left(n^3\epsilon^{-2}\beta^{-1}\ln \left(n/ \epsilon \beta \delta\right)+nd\right)$ time.
\end{theorem}

\begin{Proof}
Let the singular value decomposition of the full-rank matrix $A$ be $A = U\Sigma V^T$, with $U \in \mathbb{R}^{n \times n}$, $\Sigma \in \mathbb{R}^{n \times n}$, and $V \in \mathbb{R}^{d \times n}$; note that all the diagonal entries of $\Sigma$ are strictly positive since $A$ has full rank. We can now apply Theorem 4 of Section 6.1 of~\cite{DMMS07_FastL2_NM10} to get\footnote{We apply Theorem 4 of Section 6.1 of~\cite{DMMS07_FastL2_NM10} with $A = V^T$ and note that $\FNormS{V^T} = n \geq 1$, $\TNorm{V^T}=1$, and $\left(V^T\right)^{(i)} = V_{(i)}$.} that
\begin{equation}
\TNorm{I_n - V^TSS^TV} = \TNorm{V^TV - V^TSS^TV} \leq \epsilon
\end{equation}
for our choice of $r$ with probability at least $1-\delta$. Note that $V^TS \in \mathbb{R}^{n \times r}$ (with $r \geq n$) and let $\sigma_{i}\left(V^TS\right)$ denote the singular values of $V^TS$ for all $i \in [n]$; the above inequality implies that for all $i \in [n]$
$$\abs{1-\sigma_i^2\left(V^TS\right)}\leq \TNorm{I_n - V^TSS^TV} \leq \epsilon \leq 0.5.$$
Thus, all the singular values of $V^TS$ are strictly positive and hence $V^TS$ has full rank equal to $n$. Also, using $\epsilon \leq 0.5$,
\begin{equation}\label{eqn:possvs}
\abs{1-\sigma_i^{-2}\left(V^TS\right)}\leq \frac{\epsilon}{1-\epsilon} \leq 2\epsilon.
\end{equation}
We are now ready to prove our theorem:
\begin{eqnarray*}
\TNorm{x_{opt}-\tilde{x}_{opt}} &=& \TNorm{A^T\left(AS\right)^{\dagger T}\left(AS\right)^{\dagger}b-A^{\dagger} b}\\
&=& \TNorm{V\Sigma U^T\left(U \Sigma V^T S\right)^{\dagger T}\left(U \Sigma V^T S\right)^{\dagger}b-V\Sigma^{-1}U^T b}\\
&=& \TNorm{\Sigma U^T U \Sigma^{-1} \left(V^T S\right)^{\dagger T}\left(V^T S\right)^{\dagger}\Sigma^{-1} U^T b-\Sigma^{-1}U^T b}\\
&=& \TNorm{\left(V^T S\right)^{\dagger T}\left(V^T S\right)^{\dagger}\Sigma^{-1}U^T b-\Sigma^{-1}U^T b}.
\end{eqnarray*}
In the above derivations we substituted the SVD of $A$, dropped terms that do not change unitarily invariant norms, and used the fact that $V^TS$ and $\Sigma$ have full rank in order to simplify the pseudoinverse. Now let $\left(V^T S\right)^{\dagger T}\left(V^T S\right)^{\dagger} = I_n + E$ and note that Eqn.~(\ref{eqn:possvs}) and the fact that $V^TS$ has full rank imply
$$\TNorm{E} = \TNorm{I_n - \left(V^T S\right)^{\dagger T}\left(V^T S\right)^{\dagger}} = \max_{i \in [n]} \abs{1-\sigma_{i}^{-2}\left(V^TS\right)} \leq 2\epsilon.$$
Thus, we conclude our proof by observing that
\begin{eqnarray*}
\TNorm{x_{opt}-\tilde{x}_{opt}} &=& \TNorm{\left(I_n+E\right)\Sigma^{-1}U^T b - \Sigma^{-1}U^T b}\\
&=& \TNorm{E\Sigma^{-1}U^T b}\\
&\leq& \TNorm{E}\TNorm{\Sigma^{-1}U^T b}\\
&\leq& 2\epsilon\TNorm{x_{opt}}.
\end{eqnarray*}
In the above we used the fact that $\TNorm{x_{opt}}=\TNorm{A^{\dagger} b} = \TNorm{V\Sigma^{-1}U^T b} = \TNorm{\Sigma^{-1}U^T b}$. The running time of the algorithm follows by observing that $AS$ is an $n \times r$ matrix and thus computing its pseudoinverse takes $O(n^2r)$ time; computing $x_{opt}$ takes an additional $O(nr + dn)$ time.
\end{Proof}

\begin{algorithm}[h]
\begin{framed}

\textbf{Input:} $A \in \mathbb{R}^{n \times d}$, $b \in
\mathbb{R}^n$, error parameter $\epsilon \in (0,.5]$, failure probability $\delta$, and a set of probabilities $p_i$ (for all $i \in [d]$) summing up to one and satisfying Eqn.~(\ref{eqn:beta_underLS}).

\vspace{0.1in}

\textbf{Output:} $\tilde{x}_{opt} \in \mathbb{R}^d$.

\begin{enumerate}

\item Let $r = \frac{96n}{\beta \epsilon^2}\ln \left(\frac{96n}{\beta \epsilon^2 \sqrt{\delta}}\right)$.

\item Let $S \in \mathbb{R}^{d \times r}$ be an all-zeros matrix.

\item \textbf{For} $t=1,\ldots,r$ \textbf{do}

\begin{itemize}

\item Pick $i_t \in [d]$ such that $\Prob\left(i_t = i\right) = p_i$.

\item $S_{i_t t} = 1/\sqrt{rp_{i_t}}$.

\end{itemize}

\item Return $\tilde{x}_{opt} = A^T \left(AS\right)^{\dagger T} \left(AS\right)^{\dagger}b $.
\end{enumerate}

\end{framed}
\caption{Approximately solving under-constrained least squares problems.} \label{alg:alg_sample_fast}
\end{algorithm}

We conclude the section with a few remarks. First, assuming that $\epsilon$, $\beta$, and $\delta$ are constants and $n \ln n = o(d)$, it immediately follows that Algorithm~\ref{alg:alg_sample_fast} runs in $o(n^2d)$ time. It should be clear that we  can use Theorem~\ref{thm:main_result} and the related Algorithm~\ref{alg:Leverage_Scores} to approximate the statistical leverage scores, thus bypassing the need to exactly compute them. Second, instead of approximating the statistical leverage scores needed in Algorithm~\ref{alg:alg_sample_fast}, we could use the randomized Hadamard transform (essentially post-multiply $A$ by a randomized Hadamard transform to make all statistical leverage scores uniform). The resulting algorithm could be theoretically analyzed following the lines of~\cite{DMMS07_FastL2_NM10}. It would be interesting to evaluate experimentally the performance of the two approaches in real data. 

\subsection{Extension to streaming environments}

In this section, we consider the estimation of the leverage scores and of related
statistics when the input data set is so large that an appropriate way to
view the data is as a data stream~\cite{Muthu05}.
In this context, one is interested in computing statistics of the data
stream while making one pass (or occasionally a few additional passes) over
the data from external storage and using only a small amount of additional
space.
For an $n \times d$ matrix $A$, with $n \gg d$, small additional space
means that the space complexity only depends \emph{logarithmically} on the
high dimension $n$ and \emph{polynomially} on the low dimension $d$. When
we discuss bits of space, we assume that the entries of $A$ can be discretized
to $O(\log n)$ bit integers, though all of our results can be generalized to
arbitrary word sizes.
The general strategy behind our algorithms is
as follows.
\begin{itemize}
\item
As the data streams by, compute $TA$, for an appropriate problem-dependent
linear sketching matrix $T$, and also compute $\Pi A$, for a random
projection matrix $\Pi$.%
\footnote{In the offline setting, one would use an SRHT or another FJLT,
while in the streaming setting one could use either of the following.
If the stream is such that one sees each entire column of $A$ at once, then
one could do an FJLT on the column.
Alternatively, if one see updates to the individual entries of $A$ in an
arbitrary order, then one could apply any sketching matrix, such as those
of~\cite{Ach03_JRNL} or of~\cite{DKT10}.}
\item
After the first pass over the data, compute the matrix $R^{-1}$, as
described in Algorithm~\ref{alg:Leverage_Scores}, corresponding to $\Pi A$
(or compute the pseudoinverse of $\Pi A$ or the $R$ matrix from any other
QR decomposition of $A$).
\item
Compute $T A R^{-1} \Pi_2$, for a random projection matrix
$\Pi_2$, such as the one used by Algorithm~\ref{alg:Leverage_Scores}.
\end{itemize}
With the procedure outlined above, the matrix $T$ is effectively applied to
the rows of $AR^{-1}\Pi_2$, \emph{i.e.}, to the sketch of $A$ that
has rows with Euclidean norms approximately equal to the row norms of $U$,
and pairwise inner products approximately equal to those in $U$.
Thus statistics related to $U$ can be extracted.

\paragraph{Large Leverage Scores.}
Given any $n \times d$ matrix $A$ in a streaming setting, it is known how
to find the indices of all rows $A_{(i)}$ of $A$ for which
$\|A_{(i)}\|_2^2 \geq \tau \|A\|_F^2$, for a parameter $\tau$, and in
addition it is known how to compute a $(1+\epsilon)$-approximation to
$\|A_{(i)}\|_2^2$ for these large rows. The basic idea is to use the notion
of $\ell_2$-sampling on matrix $A$, namely, to sample random entries $A_{ij}$
with probability $A_{ij}^2/\|A\|_F^2$. A single entry can be sampled from this
distribution in a single pass using $O(\epsilon^{-2}\log^3(nd))$ bits of space \cite{MW10,AKO10-TR}.
More precisely, these references demonstrate that there is a distribution over
$O(d \epsilon^{-2}\log^3 (nd)) \times n$ matrices $T$ for which for any fixed
matrix $A \in \mathbb{R}^{n \times d}$, there is a procedure which given $TA$,
outputs a sample $(i,j) \in [n] \times [d]$ with probability
$(1 \pm \epsilon) \frac{A_{i,j}^2}{\|A\|_F^2} \pm n^{-O(1)}$. Technically, these references concern
sampling from vectors rather than matrices, so $T(A)$ is a linear operator
which treats $A$ as a length-$nd$ vector and applies the algorithm of \cite{MW10,AKO10-TR}. However,
by simply increasing the number of rows in $T$ by a factor of the small dimension $d$, we can assume
$T$ is left matrix multiplication. By considering the marginal along $[n]$, the
probability that $i = a$, for any $a \in [n]$, is
$$
(1 \pm \epsilon)\frac{\|U_{(a)}\|_2^2}{\|U\|_F^2} \pm (nd)^{-O(1)}  .
$$
By the coupon collector problem, running $O(\tau^{-1} \log \tau^{-1})$ independent
copies is enough to find a set containing all rows $A_{(i)}$ for which
$\|A_{(i)}\|_2^2 \geq \tau \|A\|_F^2$, and no rows $A_{(i)}$ for which
$\|A_{(i)}\|_2^2 < \frac{\tau}{2} \|A\|_F^2$ with probability at least $0.99$.

When applied to our setting, we can apply a random
projection matrix $\Pi$ and a linear sketching matrix $T$ which has
$O(d\tau^{-1}\epsilon^{-2} \log^3(n) \log \tau^{-1})$ rows in the following manner.
First, $TA$ and $\Pi A$ are computed in the first pass over the data;
then, at the end of the first pass, we compute $R^{-1}$; and finally, we
compute $TAR^{-1} \Pi_2$, for a random projection matrix
$\Pi_2$.
This procedure effectively applies the matrix $T$ to the rows of
$AR^{-1}\Pi_2$, which have norms equal to the row norms of $U$, up
to a factor of $1+\epsilon$. The multiplication at the end by $\Pi_2$ serves only
to speed up the time for processing $TAR^{-1}$.
Thus, by the results of~\cite{MW10,AKO10-TR}, we can find all the leverage scores
$\|U_{(i)}\|_2^2$ that are of magnitude at least $\tau \|U\|_F^2$ in small
space and a single pass over the data. By increasing the space by a factor of $O(\epsilon^{-2} \log n)$,
we can also use the $\ell_2$-samples to estimate the norms $\|U_{(i)}\|_2^2$ for the row indices
$i$ that we find.

\paragraph{Entropy.}
Given a distribution $\rho$, a statistic of $\rho$ of interest
is the entropy of
this distribution, where the entropy is defined as
$H(\rho) = \sum_i \rho(i) \log_2 ( 1/\rho(i) ) $.
This statistic can be approximated in a streaming setting.
Indeed, it is known that estimating $H(\rho)$ up to an additive $\epsilon$
can be reduced to
$(1+\tilde{\epsilon})$-approximation of the $\ell_p$-norm of the vector
$ \left( \rho(1), \ldots, \rho(n) \right)  $,
for $O(\log 1/\epsilon)$ different $p \in (0,1)$~\cite{HNO08}. Here
$\tilde{\epsilon} = \epsilon/(\log^3 1/\epsilon \cdot \log n)$.
When applied to our setting, the distribution of interest is
$\rho(i) = \frac{1}{d}\|U_{(i)}\|_2^2$.
To compute the entropy of this distribution, there exist
sketching matrices $T$ for providing $(1+\epsilon)$-approximations to the
quantity $F_p(F_2)$ of an $n \times d$ matrix $A$, where $F_p(F_2)$ is
defined as $\sum_{i=1}^n \|A_{(i)}\|_2^{2p}$, using $O(\epsilon^{-4}\log^2 n \log 1/\epsilon)$
bits of space~(see Theorem 1 of \cite{GBD08}).
Thus, to compute the entropy of the leverage score distribution, we can
do the following.
First, maintain $TA$ and $\Pi A$ in the first pass over the data, where
$T$ is a sketching matrix for $F_p(F_2)$, $p \in (0,1)$.
At the end of the first pass, compute $R^{-1}$;
and finally, compute $TAR^{-1}\Pi_2$, which effectively applies
the $F_p(F_2)$-estimation matrix $T$ to the rows of the matrix
$AR^{-1}\Pi_2$.
Therefore, by the results of~\cite{HNO08,GBD08}, we can compute an
estimate $\phi$ which is within an additive $\epsilon$ of $H(\rho)$ using
$O(d \epsilon^{-4} \log^6 n \log^{14} 1/\epsilon)$ bits of space and a single pass.
We note that it is also possible to estimate $H(\rho)$ up to a multiplicative
$1+\epsilon$ factor using small, but more, space; see, \emph{e.g.}, \cite{HNO08}.

\paragraph{Sampling Row Identities.}
Another natural problem is that of obtaining samples of rows of $A$
proportional to their leverage score importance sampling probabilities.
To do so, we use $\ell_2$-sampling ~\cite{MW10,AKO10-TR} as used
above for finding the large leverage scores. First, compute $TA$ and $\Pi A$ in the first pass over the data stream;
then, compute $R^{-1}$; and
finally, compute $TAR^{-1}$. Thus, by applying the procedures of~\cite{AKO10-TR} a total of $s$
times independently, we obtain $s$ samples $i_1, \ldots, i_s$, with
replacement, of rows of $A$ proportional to
$\|U_{(i_1)}\|_2^2, \ldots, \|U_{(i_s)}\|_2^2$, \emph{i.e.}, to
their leverage score.
The algorithm requires $O(sd \epsilon^{-2} \log^4 n)$ bits of space and runs in a
single pass.
To obtain more than just the row identities $i_1, \ldots, i_s$, \emph{e.g.},
to obtain the actual samples, one can read off these rows from $A$ in a
second pass over the matrix.


\end{document}